\documentclass[12pt,reqno,a4paper,pagebackref]{amsart}
\usepackage{graphicx}
\usepackage{fancyhdr}
\usepackage{enumerate}
\usepackage{amsmath}

\usepackage[multiple]{footmisc}

\usepackage{amsthm}
\usepackage{mathabx}
\usepackage{amssymb}
\usepackage{relsize}
\usepackage{pdfsync}
\usepackage[colorlinks=true,linkcolor=blue,citecolor=magenta]{hyperref}

\usepackage[normalem]{ulem}

\usepackage{xcolor}

\usepackage{cancel}

\usepackage{calligra}

\usepackage{tikz}

\usepackage{changebar}
\usepackage{pdfsync}
\usetikzlibrary{decorations.pathmorphing,patterns}
\usetikzlibrary{calc,arrows}

\usepackage{pgf}

\usepackage{setspace}

\usetikzlibrary{automata}
\usetikzlibrary{positioning}

\tikzset{
    state/.style={
           rectangle,
           rounded corners,
           draw=black, very thick,
           minimum height=2em,
           inner sep=2pt,
           text centered,
           },
}

\usepackage{MnSymbol}

\usepackage{changebar}
\usepackage{pdfsync}

\usepackage[normalem]{ulem}

\usepackage{mathtools}

\usepackage[
top    = 1.75cm,
bottom =4.00cm,
left   = 3.00cm,
right  = 2.50cm]{geometry}

\makeindex
\newtheorem{theorem}{Theorem}[section]

\newtheorem{proposition}[theorem]{Proposition}
\newtheorem{corollary}[theorem]{Corollary}

\newtheorem{definition}[theorem]{Definition}
\newtheorem{remark}[theorem]{Remark}

\newcommand{\cal}{\mathcal}
\newcommand{\mc}[1]{{\mathcal #1}}

\newcommand\bbR{{\mathbb R}}
\newcommand\bbZ{{\mathbb Z}}

\newcommand\al{\alpha}

\newcommand\lang{\langle\langle}
\newcommand\rang{\rangle\rangle}

\newcommand\bbE{{\mathbb E}}

\newcommand\R{{\mathbb R}}

\newcommand \ga{\gamma}
\newcommand \om{\omega}
\newcommand \la{\lambda}

\newcommand\Om{\Omega}

\renewcommand{\ge}{\geqslant}
\renewcommand{\le}{\leqslant}
\newcommand{\dd  }{\mathrm{d}}

\renewcommand{\hat}{\widehat}
\renewcommand{\tilde}{\widetilde}
\renewcommand{\bar}{\overline}
\numberwithin{equation}{section}

\newcommand{\pv}{{\bf p}}
\newcommand{\qv}{{\bf q}}

\newcommand{\cF}{{\mathcal F}} 

\begin{document}

\title[Heat equation with boundaries]{Heat flow in a
  periodically forced, thermostatted chain II}

 \author{Tomasz Komorowski}
 \address{Tomasz Komorowski\\Institute of Mathematics,
   Polish Academy
   Of Sciences\\Warsaw, Poland \\
  \emph{and }  Institute of Mathematics, Maria Curie-Sk\l odowska University\\
   Lublin, Poland.} 
\email{{\tt tkomorowski@impan.pl}}

\author{Joel L. Lebowitz}
\address{Joel Lebowitz, Departments of Mathematics and Physics,  Rutgers University}
\email{lebowitz@math.rutgers.edu}
 
 \author{Stefano Olla}
 \address{Stefano Olla, CEREMADE,
   Universit\'e Paris Dauphine - PSL Research University \\
\emph{and}  Institut Universitaire de France\\
\emph{and} GSSI, L'Aquila}
  \email{olla@ceremade.dauphine.fr}


\begin{abstract}
  We derive a macroscopic heat equation for the temperature of
  a pinned harmonic chain subject
  to a periodic force at its right side and in contact with a heat bath at its left side.
  The microscopic dynamics in the bulk is given by
  the Hamiltonian equation of motion plus a reversal of the velocity of a particle
occurring independently for each particle at exponential times, with rate $\gamma$.
The latter produces a finite heat conductivity. Starting with an initial
probability distribution for a chain of $n$ particles we compute 
the current and the  local temperature given by the
expected value of the local energy. Scaling space and time diffusively yields, in the $n\to+\infty$ limit,
the heat equation for the macroscopic temperature profile $T(t,u),$
$t>0$, $u \in [0,1]$. It is to be solved for initial conditions
$T(0,u)$ and specified $T(t,0)=T_-$, the temperature
of the left heat reservoir and a fixed heat flux $J$,
entering the system at $u=1$. $|J|$ equals the work done by the
periodic force which is computed explicitly for each $n$.
\end{abstract}



\maketitle


\section{Introduction}
\label{intro}

The emergence of the heat equation from a microscopic dynamics after a diffusive
rescaling of space and time is a challenging mathematical problem
in non-equilibrium statistical mechanics \cite{bll00}.
Here we study this problem in the context of conversion of work into heat
in a simple model: a pinned harmonic chain.
The system is in contact at its left end with a thermal reservoir
at temperature $T_-$ which acts
on the leftmost particle via a Langevin force (Ornstein-Uhlenbeck process).
The rightmost particle is acted on by a deterministic periodic force
which does work on the system.
The work pumps energy into the system with the energy then
flowing into the reservoir in the form of heat.

To describe this flow we need to know the heat conductivity of the system.
As {it} is well known, the harmonic crystal has an infinite heat conductivity \cite{RLL67}.
To model realistic systems with finite heat conductivity we add to the harmonic dynamics
a random velocity reversal.
{It} models in a simple way the various dissipative mechanisms in real systems
and produces a finite conductivity (cf. \cite{bo1}, \cite{bbjko}).

In paper \cite{klo22}, which is the first part of the present work, we studied this system in the limit
$t\to\infty$, {see Section \ref{sec-cmp} for rigorous statements
  of the main
results obtained there.}
In this limit the probability distribution
{ of the phase space configurations} is periodic with the
period of the external force, see Theorem \ref{periodic} below. We also showed that with a proper scaling
of the force and period the averaged temperature
profile satisfies the stationary heat
equation with an  explicitly given heat current. 
In the present paper we study the time dependent evolution of the system,
on the diffusive
time scale, starting with some specified initial distribution.
We derive a heat equation for the temperature profile of the system.

{  The periodic forcing generates a Neumann type  boundary condition for the
macroscopic heat equation, so that the  gradient of the temperature
at the boundary must satisfy Fourier law with the boundary energy current
generated by the work of the periodic forcing (see \eqref{eq:5} below).
On the left side the boundary condition is given by the assigned value
$T_-$, the temperature of the heat bath. 
As $t\to\infty$ the profile   converges to the  macroscopic profile  obtained in  \cite{klo22}. 

The energy diffusion in the harmonic chain on a finite lattice,
  with energy conserving noise and
  Langevin heat bath at different temperatures at the boundaries, have been previously
  considered \cite{bo05,bkll11,bkll12,kos3,lukk}.
But complete mathematical results, describing the time evolution of
the macroscopic  temperature  profile, have been obtained only for unpinned chains
\cite{bo05,kos3}.

This article gives the first proof of the heat equation for the pinned chain in a finite interval,
  and the method can be applied with different boundary conditions
  (see Remark \ref{rem-mixed}).
  {Investigation about energy transport in anharmonic chain
    under periodic forcing can be found in \cite{GL02}, \cite{JKLA09},
    and very recently in \cite{PBS22}}.
  In the review article \cite{klos22} we considered various extensions
  of the present results to
  unpinned, multidimensional and anharmonic dynamics.

\subsection{Structure of the article}
\label{sec:structure-article}

{We start Section \ref{sec:description} with the  precise
  description of  the  
dynamics of the oscillator chain. 
Then, as already mentioned, in  Section \ref{sec-cmp} } we give an account
of results obtained in \cite{klo22}. In Section
\ref{sec:state-results} we formulate our two main theorems:
\autoref{th4} about the limit current generated at the boundary by a periodic force, and
\autoref{th1} about the convergence of the {energy  profile} to the solution of
the heat equation with mixed boundary conditions.}

In Section \ref{sec:bound-energ}
we obtain a uniform bound on the total energy at any macroscopic time by
an entropy argument. As a corollary (cf. Corollary \ref{lem:bound2})
we obtain a uniform bound on
the time integrated energy current, with respect to the size of the system. 

Section \ref{sec:equip-energy-fluct} contains the proof of the equipartition of energy:
Proposition \ref{cor012912-21} shows that the limit profiles
of the kinetic and potential energy
are equal. {Furthermore, we show there}
the \emph{fluctuation-dissipation relation} (\eqref{eq:33}). It gives an exact
decomposition of the energy currents into a dissipative term
(given by a gradient of a local function) and a fluctuation term (given by
the generator of the dynamics applied to a local function).

The fluctuation-dissipation relation \eqref{eq:33} and equipartition of energy
\eqref{011312-21} are two of the ingredients for the proof of the main
\autoref{th1}. The third {component} is a local equilibrium result
for the limit covariance of the positions integrated in time. {It is
formulated} in Proposition \ref{prop010803-22},
for the covariances in the bulk, and in Proposition \ref{prop-boundaryeq},
for the boundaries. The local equilibrium {property} allows to identify correctly
the thermal diffusivity in the proof of \autoref{th1}, see \autoref{sec:proofth1}.

{The technical part of the argument is presented in the appendices:
  the proof of the local equilibrium  is given  \autoref{sec-proofs-le},
after the analysis of the time evolution of the {matrix for
the time integrated covariances of  positions and momenta}, carried out 
in \autoref{ev-covariance}. Both in  \autoref{ev-covariance} and  \autoref{sec-proofs-le}
we use results proven in \cite{klo22}, when possible.
\autoref{app-means} contains the proof of the current
asymptotics (\autoref{th4}),
that involves only the dynamics of the averages of the
configurations. Appendix \ref{A:unique} contains the proof of the uniqueness of
measured  valued
solutions of the Dirichlet-Neumann initial-boundary problem for the
heat equation, satisfied by the  limiting energy profile. Finally, in
Appendix \ref{secF} we present an argument for the relative entropy
inequality stated in Proposition \ref{cor021211-19}.
}

\section{Description of the model}
\label{sec:description}

We consider a pinned chain of $n+1$-harmonic oscillators in contact on the left with
a Langevin heat bath at temperature $T_-$, and with a periodic force acting on
the last particle on the right.
The configuration of particle positions and momenta are specified by
\begin{equation}
  \label{eq:1}
  (\mathbf q, \mathbf p) =
  (q_0, \dots, q_n, p_0, \dots, p_n) \in \Om_n:=\R^{n+1}\times\R^{n+1}. 
\end{equation}
We should think of the positions $q_x$ as relative displacement from a point,
say $x$ in a finite lattice $\{0,1,\ldots,n\}$.
The total energy of the chain is given  by the Hamiltonian:
$\mathcal{H}_n (\mathbf q, \mathbf p):=
\sum_{x=0}^n {\cal E}_x (\mathbf q, \mathbf p),$
where the energy of particle $x$ is defined  by
\begin{equation}
\label{Ex}
{\cal E}_x (\mathbf q, \mathbf p):=  \frac{p_x^2}2 +
\frac12 (q_{x}-q_{x-1})^2 +\frac{\om_0^2 q_x^2}{2} ,\
\quad x = 0, \dots, n,
\end{equation}
where $\om_0>0$ is   the pinning strenght.
We adopt the convention that $q_{-1}:=q_0$.

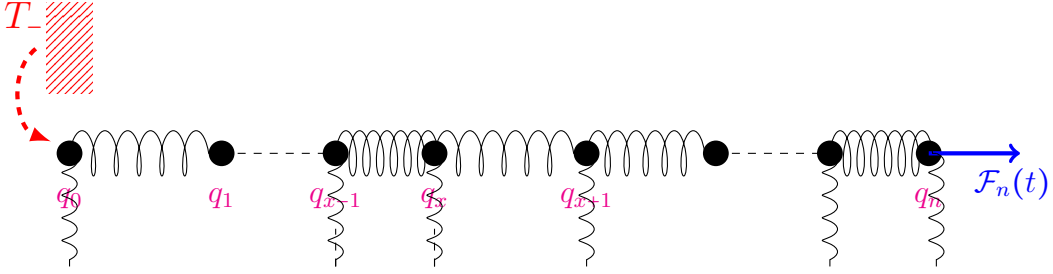
\begin{figure}[h!]

\begin{center}
\begin{tikzpicture}[scale = 1.0]
\node[circle,fill=black,inner sep=1.2mm] (e) at (0,0) {};
\node[circle,fill=black,inner sep=1.2mm] (f) at (2,0) {};
\node[circle,fill=black,inner sep=1.2mm] (g) at (3.5,0) {};
\node[circle,fill=black,inner sep=1.2mm] (h) at (4.8,0) {};
\node[circle,fill=black,inner sep=1.2mm] (i) at (6.8,0) {};
\node[circle,fill=black,inner sep=1.2mm] (j) at (8.5,0) {};
\node[circle,fill=black,inner sep=1.2mm] (k) at (10,0) {};
\node[circle,fill=black,inner sep=1.2mm] (l) at (11.3,0) {};

\draw[dashed] (2,0) -- (3.5,0);
\draw[dashed] (8.5,0) -- (10,0);
\draw[ultra thick, blue, ->] (11.3,0) -- (12.5,0);


\draw (0, -0.6) node[] {\color{magenta} $q_{0}$};
\draw (2, -0.6) node[] {\color{magenta} $q_{1}$};
\draw (11.3, -0.6) node[] {\color{magenta} $q_{n}$};
\draw (3.5, -0.6) node[] {\color{magenta} $q_{x-1}$};
\draw (4.8, -0.6) node[] {\color{magenta} $q_{x}$};
\draw (6.8, -0.6) node[] {\color{magenta} $q_{x+1}$};
\draw[dashed] (3.5,-1.5) -- (3.5,-1);
\draw[dashed] (4.8,-1.5) -- (4.8,-1);

\draw (-0.6,1.8) node[] {\large\color{red}$T_-$};
\draw (12.4,-0.4) node[] {\color{blue}${\mathcal F}_n(t)$};

\draw[decoration={aspect=0.3, segment length=3mm, amplitude=3mm,coil},decorate] (0,0) -- (2,0);
\draw[decoration={aspect=0.3, segment length=3mm, amplitude=1mm,coil},decorate] (0,0) -- (0,-1.5);
\draw[decoration={aspect=0.3, segment length=1.8mm, amplitude=3mm,coil},decorate] (3.5,0) -- (4.9,0);
\draw[decoration={aspect=0.3, segment length=3mm, amplitude=1mm,coil},decorate] (3.5,0) -- (3.5,-1.5);
\draw[decoration={aspect=0.3, segment length=3mm, amplitude=3mm,coil},decorate] (4.8,0) -- (6.9,0);
\draw[decoration={aspect=0.3, segment length=3mm, amplitude=1mm,coil},decorate] (4.8,0) -- (4.8,-1.5);
\draw[decoration={aspect=0.3, segment length=2.5mm, amplitude=3mm,coil},decorate] (6.8,0) -- (8.6,0);
\draw[decoration={aspect=0.3, segment length=3mm, amplitude=1mm,coil},decorate] (6.8,0) -- (6.8,-1.5);
\draw[decoration={aspect=0.3, segment length=1.8mm, amplitude=3mm,coil},decorate] (10,0) -- (11.4,0);
\draw[decoration={aspect=0.3, segment length=3mm, amplitude=1mm,coil},decorate] (10,0) -- (10,-1.5);
\draw[decoration={aspect=0.3, segment length=3mm, amplitude=1mm,coil},decorate] (11.4,0) -- (11.4,-1.5);

\fill [pattern = north east lines, pattern color=red] (-0.3,0.8) rectangle (0.3,2);
\node (c) at (-0.3,1.5) {};
\node (d) at (-0.1,0.1) {};
 \node (a) at (11.6,1.5) {};
 \node (b) at (11.4,0.1) {};

\draw (c) edge[dashed, ultra thick, red, ->, >=latex, bend right=60] (d);

\end{tikzpicture}
\end{center}
\end{figure}

The  microscopic dynamics of
the process $\{(\mathbf q(t), \mathbf p(t))\}_{t\ge0}$
describing the total chain is  given in the bulk by
\begin{equation} 
\label{eq:flip}
\begin{aligned}
  \dot   q_x(t) &= p_x(t) ,
  \qquad x\in \{0, \dots, n\},\\
  \dd   p_x(t) &=  \left(\Delta q_x(t)-\om_0^2 q_x(t)\right) \dd t-   2
  p_x(t-) \dd N_x(\gamma t),
  \quad x\in \{1, \dots, n-1\}
  \end{aligned} \end{equation}
and at the boundaries by 
\begin{align}
     \dd   p_0(t) &=   \; \Big(q_1(t)-q_0(t) - \om_0^2 q_0(t) \Big) \dd   t
                     -
                    2  \gamma p_0(t) \dd t
                    +\sqrt{4  \gamma T_-} \dd \tilde w_-(t),
                    \vphantom{\Big(} \label{eq:pbdf} \\
  \dd   p_n(t) &=  \; \Big(q_{n-1}(t) -q_n(t) -\om_0^2 q_n(t) \Big)  \dd   t  +\cF_n(t)
                 \dd t - 2  p_n(t-) \dd N_n(\gamma t).
                     \vphantom{\Big(}  
\notag
\end{align}
Here $\Delta$ is the Neumann discrete laplacian, corresponding to the choice
$q_{n+1}:=q_n$ and {$q_{-1}= q_0$}, {see \eqref{012901-23} below.} Processes $\{N_x(t), x=1,\ldots,n\}$
are independent Poisson of intensity $1$, while
$\tilde w_-(t)$ is a standard one dimensional Wiener process,
independent of the Poisson processes.
Parameter $\gamma>0$ 
regulates the intensity of the random perturbations
and the Langevin thermostat.
{We have choosen the same parameter in order to simplify notations,
  it does not affect the results {concerning} the macroscopic properties of the dynamics.}

We assume that the forcing $\cF_n(t)$ is given by
\begin{equation}
\label{Fnt}
 \cF_n(t)= \;\frac 1{\sqrt n} \cF\left(\frac{t}{ \theta}\right).
\end{equation} 
 where $ \cF(t)$ is a $1$-periodic function
such that
{\begin{equation}
  \label{eq:2}
  \int_0^1  \cF(t) \dd t = 0, \qquad  \int_0^1  \cF(t)^2 \dd t >
  0\quad\mbox{and}\quad \sum_{\ell\in\bbZ} |\widehat     \cF(\ell)|<+\infty.
\end{equation}
Here
\begin{equation}
\label{cF}
\hat    \cF(\ell)=\int_0^1 e^{-2\pi i\ell t }\cF(t)\dd t,\quad {\ell\in\bbZ},
\end{equation}
are the Fourier coefficients of the force.
Note that by \eqref{eq:2} we have $\hat   \cF(0)=0$.}



{For a given function $f:\{0,\ldots,n\}\to\bbR$ define the Neumann laplacian 
 \begin{equation}
   \label{NL}
   \Delta f_x:=f_{x+1}+f_{x-1}-2f_x,\,x=0,\ldots,n,
 \end{equation}
 with the convention $f_{-1}:=f_0$ and $f_{n+1}:=f_n$.}
The generator of the dynamics can be then written as
\begin{equation}
  \label{eq:7}
  \mathcal G_t =  \mathcal A_t +  \gamma S_{\text{flip}}
  + 2   \gamma S_-,
\end{equation}
where
\begin{equation}
  \label{eq:8}
  \mathcal A_t = \sum_{x=0}^n p_x \partial_{q_x}
  + \sum_{x=0}^n  (\Delta q_{x}-\om^2_0q_x) \partial_{p_x}
  +  \cF_n(t)  \partial_{p_n}
\end{equation}
and
\begin{equation}
  \label{eq:21}
   S_{\text{flip}} F (\qv,\pv) = \sum_{x=1}^{n}   \Big( F(\qv,\pv^x) - F(\qv,\pv)\Big),
 \end{equation}
 Here $F:\bbR^{2(n+1)}\to\bbR$ is a bounded and measurable function,
 $\pv^x$ is the velocity configuration with sign flipped at the
 $x$ component, i.e. $\pv^x=(p_0^x,\ldots,p_n^x)$, with $p_y^x=p_y$,
 $y\not =x$ and  $p_x^x=-p_x$. Furthermore,
 \begin{equation}
   \label{eq:10}
   S_- = T_- \partial_{p_0}^2 - p_0 \partial_{p_0}.
 \end{equation}
 
The {microscopic} energy currents are given by
\begin{equation}
\label{eq:current}
  \mathcal G_t \mathcal E_x(t)  = j_{x-1,x}(t) - j_{x,x+1}(t) ,
\end{equation} 
{with $\mathcal E_x(t):={\cal E}_x \big(\mathbf q(t), \mathbf
p(t)\big)$ and }
$$
j_{x,x+1}(t):=- p_x(t) \big(q_{x+1}(t) - q_x(t)\big) , \qquad \mbox{if }\quad x \in
\{0,...,n-1\}
$$  
and at the boundaries 
  \begin{equation} \label{eq:current-bound}
     j_{-1,0} (t):= 2 { \gamma} \left(T_- - p_0^2(t) \right),
      \qquad
  j_{n,n+1} (t):=    -  \cF_n(t)    p_n(t).
\end{equation}

\subsection{Summary of  results concerning periodic stationary state}

\label{sec-cmp}

{The present section is devoted to presentation of the results of
\cite{klo22} (some additional facts are contained in   \cite{klo22a}). They concern the case when the
chain is in its (periodic) stationary state. More precisely,
we say that the family of probability measures
$\{\mu_t^P, t\in[0,+\infty)\}$ constitutes a \emph{periodic stationary state}
  for  the chain described by \eqref{eq:flip} and \eqref{eq:pbdf} if
  it is
  a solution of the forward equation: { for any function $F$ in the domain of
    $\mathcal G_t$:
    \begin{equation}
      \partial_t \int F(\qv,\pv) \mu_t^P(\dd\qv,\dd\pv)
      = \int (\mathcal G_t F(\qv,\pv)) \mu_t^P(\dd\qv,\dd\pv),\label{eq:fwe}
\end{equation}
  }
such that $\mu_{t + \theta}^P=\mu_{t}^P$. }

{Given a measurable function $F:\bbR^{2(n+1)}\to\bbR$ we denote
\begin{equation}
\label{bar}
\lang F\rang:= \frac{1}{\theta}\int_0^{\theta}\dd t\int_{\bbR^{2(n+1)}}F(\qv,\pv)\mu_t^P(\dd\qv,\dd\pv),
\end{equation}
provided that  $|F(\qv,\pv)|$ is integrable
w.r.t. the respective  product measure.}

{It has been shown, see \cite[Theorem 1.1, Proposition A.1]{klo22}  and
also \cite[Theorem A.2]{klo22a}, that there exists a
unique periodic, stationary state. 
\begin{theorem}
\label{periodic}
For a fixed $n\ge1$ there exists a unique periodic stationary state
$\{\mu_s^P,s\in[0,+\infty)\}$ for the system \eqref{eq:flip}-\eqref{eq:pbdf}.
 The measures $\mu_s^P$ are absolutely continuous with respect to the Lebesgue
 measure $\dd\qv\dd\pv$ and the respective densities
 $\mu_s^P(\dd\qv,\dd\pv)=f_s^P(\qv,\pv) \dd\qv\dd\pv$ are
 strictly positive.
 {The time averages of all the second moments $\lang
   p_xp_y\rang$, $\lang
   p_xq_y\rang$  and $\lang
   q_xq_y\rang$ are finite and $\min_x \lang p_x^2\rang \ge T_-$.
 Furthermore, given an arbitrary initial probability distribution
 $\mu$ on $\bbR^{2(n+1)}$ and $(\mu_t)$
{the solution of \eqref{eq:fwe}}
such that $\mu_{0}=\mu$,  we have
\begin{equation}
    \label{012301-23a}
    \lim_{t\to+\infty}\|\mu_{t}-\mu_t^P\|_{\rm TV}=0.
 \end{equation}
 Here $\|\cdot\|_{\rm TV}$ denotes the total
  variation norm.}
\end{theorem}}

{In the periodic stationary state   the time averaged
energy current $  J_n=\lang j_{x,x+1}\rang$ is constant for
$x=-1,\ldots,n$.
In particular 
\begin{equation}  \label{eq:38a}
  J_n=
 -  \frac {1}{\sqrt{n}\theta}  \int_0^{\theta}  \cF\left(\frac{s}{\theta}\right)
 \bar p_n(s)\dd s    ,
\end{equation}
where $\bar p_x(s)  :=  \int_{\bbR^{2(n+1)}}p_x
\mu_s^P(\dd\qv,\dd\pv)$. It turns out that the stationary current is
of size $O(1/n)$ as can be seen from the following.
\begin{theorem}[see Theorem 3.1 of \cite{klo22}]
  \label{thm-current}
{Suppose that ${\cal F}(\cdot)$ satisfies \eqref{eq:2} and, in
  addition, we also have $\sum_{\ell\in\bbZ}\ell^2|\hat{\cal F}(\ell)|^2<+\infty$.}
Then, 
\begin{equation}
\label{051021-05z}
\lim_{n\to+\infty}n J_n=  J
:=-\left(\frac {2\pi}{\theta}\right)^2   \sum_{\ell\in\bbZ}
\ell^2{\cal Q} (\ell), 
\end{equation}
with ${\cal Q}(\ell)$ given  by, 
\begin{equation} 
\label{021205-21f1z}
\begin{aligned}
  & {\cal Q}(\ell)= 4\gamma|\hat     \cF(\ell)|^2
  \int_0^1 \cos^2\left(\frac{\pi z}{2}\right)
  \left\{\left[4\sin^2\left(\frac{\pi z}{2}\right)
      +\om_0^2 -\left(\frac{2\pi\ell}{\theta}\right)^2\right]^2
    +\left(\frac{{4} \gamma \pi \ell}{\theta}\right)^2 \right\}^{-1}\dd z.
\end{aligned} \end{equation}
In the more general case when the forcing $\cF_n(t)$ is $\theta_n$-periodic, with the
period $\theta_n=n^b\theta$ and the amplitude $n^{a}$, i.e.
$
 \cF_n(t)= \;n^{a} \cF\left(\frac{t}{ \theta_n}\right),
$
and
\begin{equation}
\label{a-bz}
b-a=\frac{1}{2},\quad a\le 0\quad\mbox{and}\quad b>0 
\end{equation}
the convergence in \eqref{051021-05z} still holds. However, then
\begin{equation} 
\label{021205-21f2z}
\begin{aligned}
   {\cal Q} (\ell)= 4\gamma|\hat     \cF(\ell)|^2
   \int_0^1 \cos^2\left(\frac{\pi z}{2}\right)
   \left[4\sin^2\left(\frac{\pi z}{2}\right)
       +\om_0^2 
     \right]^{-2}\dd z
,\quad\mbox{when }b>0.
\end{aligned} 
\end{equation}
\end{theorem}}

{Concerning the convergence of the energy profile we have shown the following, see
\cite[Theorem 3.4]{klo22}.
\begin{theorem}\label{th1z}
Under the assumptions of Theorem \ref{thm-current}
  we have {
  \begin{equation}
    \label{eq:3z}
    \lim_{n\to\infty} \frac 1n \sum_{x=0}^n \varphi\left(\frac x{n+1} \right) \lang p^2_x\rang
    {=  \lim_{n\to\infty} \frac 1n \sum_{x=0}^n \varphi\left(\frac x{n+1} \right)
    \lang {\cal E}_x\rang}
    = \int_0^1 \varphi (u) T(u) \dd u,
  \end{equation}
 } with 
  \begin{equation}
    \label{eq:5z}
      T(u) = T_--\frac{4\gamma Ju}{D}, \quad u\in[0,1],
  \end{equation}
for any $\varphi\in C[0,1]$.
  Here $J$ is given by \eqref{051021-05z}
  and
  \begin{equation}
    D = 1 - \omega_0^2 \Big(G_{\om_0}(0)+ G_{\om_0}(1)\Big) =\frac{2}{2+\om_0^2+\om_0\sqrt{\om_0^2+4}},
    \label{eq:13}
  \end{equation}
  where $G_{\om_0}(\ell)$ is the Green function defined in \eqref{GR}. 
\end{theorem}}

{Concerning the time variance of the average kinetic energy we have
shown the following.
\begin{theorem}[Theorem 9.1, \cite{klo22}]
  \label{thm:var}
Suppose that the forcing ${\cal F}_n(\cdot)$ is given by
\eqref{Fnt}, where ${\cal F}(\cdot)$ satisfies the hypotheses made in
Theorem \ref{thm-current}. Then, there exists a constant
$C>0$ such that  
  \begin{equation}
    \label{eq:91}
  {    \sum_{x=0}^n \frac 1\theta
    \int_0^\theta \left(\bar{p_x^2}(t) - \lang p_x^2\rang \right)^2
    \dd t  
  \le \frac{C}{n^2},\quad n=1,2,\ldots.}
\end{equation}
Here
$
\bar{p_x^2}(t) :=\int_{\bbR^{2(n+1)}}p_x^2
\mu_t^P(\dd\qv,\dd\pv)$. 
\end{theorem}}

\subsection{Statements of the main results}
\label{sec:state-results}

\subsubsection{Macroscopic energy current due to work}
\label{sec2.1.1}
{ The first results concerns the work done by the forcing in  a diffusive limit,
i.e.
\begin{equation}
  \label{eq:16}
  J_n(t, \mu) = \frac 1n \int_0^{n^2t} \bbE_{\mu}\left(  j_{n,n+1} (s, \qv(s),\pv(s))  \right) ds
  =  - \frac 1n \int_0^{n^2t} \cF_n(s) \bbE_{\mu}\left( p_n(s) \right) ds,
\end{equation}
where $ \bbE_{\mu}$ denotes the expectation of the process with the
initial configuration $(\qv,\pv)$ distributed according to a 
probability measure $\mu$. We shall write $J_n(t, \bf{q},\bf{p})$ if
for a deterministic initial configuration $(\bf{q},\bf{p})$, i.e. 
$\mu=\delta_{\bf{q},\bf{p}}$, the $\delta$-measure that gives probability
1 to such configuration.

Assume furthermore that $(\mu_n)$ is a
   sequence of initial distributions, with each $\mu_n$ probability
   measure on $\bbR^{n+1}\times \bbR^{n+1}$. We suppose that
    there exist $C>0$ and
   $\delta\in[0,2)$ for
   which   for any integer $n\ge 1$
   \begin{equation}
  \label{fHn}
      \mathcal H_n(\bar\qv_n, \bar\pv_n) \le Cn^\delta. 
    \end{equation}
    Here $(\bar\qv_n, \bar\pv_n)$ {is} the vector of the averages of the
    configuration with respect to $\mu_n$.}
  {We are interested essentially in the case $\delta =1$, but
    \autoref{th4} is valid also for any $\delta <2$.} { In Proposition \ref{prop011012-21} we prove
    that, in the diffusive time scaling, the energy due to the averages \eqref{fHn}
    becomes negligible at any time $t>0$}.
    
{In Section \ref{sec-th4} of the Appendix we prove the following.}
 \begin{theorem} \label{th4}
 {Under the assumptions listed above, we have
       \begin{equation}
         \label{Jnn}
         \lim_{n\to+\infty}\sup_{t\ge 0}\Big|J_n(t, \mu_n)-Jt\Big|=0,       
      \end{equation}}
where $J$ is given by \eqref{051021-05z}.
\end{theorem}

\begin{remark}
  The asymptotic current $J$ is the same as in the stationary state (cf. \cite{klo22})
  and it does not depend on the initial configuration.
\end{remark}
\begin{remark}
{ Analogously to  the stationary case, rescaling the period
  $\theta$ with $n$ and the strenght of the force in such a way that
\begin{equation}
\label{Fnta-b}
\cF_n(t)= \;n^{a} \cF\left(\frac{t}{ n^{b}\theta}\right),
\qquad b-a=\frac{1}{2},\quad a\le 0\quad\mbox{and}\quad b>0,
\end{equation}
\autoref{th4} still holds, but with a different value of the
current. Namely, $J$  is given by
\eqref{051021-05z} with ${\cal Q}(\ell)$ defined by  \eqref{021205-21f2z}.
Formula \eqref{021205-21f2z} corresponds to \eqref{021205-21f1z} 
with the value $\theta = \infty$. If $b-a\not=1/2$ the macrosopic current $nJ_n$
is not of order $O(1)$, which leads to an anomalous behavior of the heat
conductivity of the chain (it vanishes, if $b-a>1/2$, becomes
unbounded, if $b-a<1/2$. The assumption $a\le 0$ guarantees that the
force acting on the system does not become infinite, as $n\to+\infty$.}
\end{remark}
\begin{remark}
{Using  contour integration  it is possible to calculate  the
quantities appearing  in  \eqref{021205-21f1z}  and
\eqref{021205-21f2z}, see \cite[Appendix D]{klo22a}. In the case of
\eqref{021205-21f1z}  we obtain  
  \begin{align*}
     &{\cal Q}(\ell) =\frac{\theta |\hat    \cF(\ell)|^2}{2\pi\ell}
       {\rm Im}\left( \left\{\frac{2 }{\la(\om_0,\ell)\sqrt{1+4/\la(\om_0,\ell)}}
       +\frac12\right\} \left\{1+\frac{ \la(\om_0,\ell)}{2}\Big(1+\sqrt{1
     +\frac{4}{\la(\om_0,\ell)}}\Big)\right\}^{-1}\right),
\end{align*}
with
$$
\la(\om_0,\ell):=\om_0^2 -\left(\frac{2\pi\ell}{\theta}\right)^2 
    +i\left(\frac{{4} \gamma \pi \ell}{\theta}\right) .
    $$
    Furthermore,  in the case of
\eqref{021205-21f2z} we have 
    \begin{align*}
    {\cal Q}(\ell)
    = 
    \frac{2\gamma|\hat     \cF(\ell)|^2  (4 +\om_0^2) }
    { (\om_0^4+4 \om_0^2+8)^{3/2}} .
\end{align*}}
  \end{remark}

\subsubsection{Macroscopic energy profile}

Let $\nu_{T_-} (\dd{\bf q},\dd{\bf p})$ be defined as the product
  Gaussian measure on $\Omega_n$ {(see \eqref{eq:1})}
 of {zero} average and variance $T_->0$ given by
\begin{equation} \label{eq:nuT}
  \begin{split}
    &\nu_{T_-} (\dd{\bf q},\dd{\bf p}) : =
    \frac{1}{Z}\prod_{x=0}^n \exp\left\{-\mc E_x({\bf
      q},{\bf p})/T_- \right\} \dd{\bf q}\dd{\bf p},
\end{split}
\end{equation}
where $Z$ is the normalizing constant.
Let $f(\qv,\pv)$ be a probability density  
with respect
to $\nu_{T_-}$. We denote the relative entropy
 \begin{equation}
  \label{eq:7a}
\mathbf{H}_{n}(f) :=  \int_{\Om_n} f(\qv,\pv)\log f(\qv,\pv) \dd \nu_{T_-}(\qv,\pv).
\end{equation}
We assume now
that the initial distribution {$\mu_n$} has  density $f_n(0, \qv,\pv)$, with respect
to $\nu_{T_-}$, such that {there exists a constant $C>0$ for which  } 
\begin{equation}
\label{eq:ass0}
    {\mathbf{H}}_n(f_n(0)) \le C n,\quad n=1,2,\ldots.
  \end{equation}
  {For example, it can be verified that local Gibbs measures of the form
    \begin{equation}
      \label{eq:LGM}
      f_n(\qv,\pv) \dd \nu_{T_-}(\qv,\pv) =
      \prod_{x=0}^n \exp\left\{-\frac{\mc E_x({\bf q},{\bf p})}{T_{x,n}} \right\} \dd{\bf q}\dd{\bf p},
    \end{equation}
with $\inf_{x,n}  T_{x,n} >0$ satisfy \eqref{eq:ass0}.}
 {
   At this point we only remark that, due to the
   entropy inequality (see the proof of Corollary \ref{th2} below),
   assumption \eqref{eq:ass0} implies 
   \begin{equation*}
  \sup_{n\ge 1}
  \bbE_{\mu_n}\left[ \frac 1{n+1} \sum_{x=0}^n {\cal E}_x(0)\right] <+\infty.
\end{equation*}
Furthermore, since the Hamiltonian ${\cal H}(\cdot,\cdot)$ is a convex
function, by the Jensen inequality
$$
\sup_{n\ge 1}\frac 1{n+1}\mathcal H_n(\bar\qv_n, \bar\pv_n) \le \sup_{n\ge 1}
  \bbE_{\mu_n}\left[ \frac 1{n+1} \sum_{x=0}^n \mathcal H_n(\qv,
\pv)\right]<+\infty,
$$
so \eqref{fHn} is satisfied with $\delta=1$.}

{ Denote by $\mathcal M_{\rm fin}([0,1])$, resp $\mathcal M_{+}([0,1])$  the space of bounded variation,
  {Borel}, resp. positive, measures on the interval $[0,1]$
  endowed with the weak topology.
Before formulating the main result we introduce the notion of a
measured valued solution of the following initial-boundary value
problem
 \begin{equation}
    \label{eq:5}
    \begin{split}
      &\partial_t T  = \frac{D}{4\gamma} \partial_u^2 T, \quad u\in(0,1),\\
      &T(t,0) = T_-, \quad  \partial_u T (t,1) = -\frac{4\gamma   J}{D},\quad
      T(0,\dd u)  = T_0(\dd u).
    \end{split}
  \end{equation}
 Here   $J$ and $D$  are defined by    \eqref{051021-05z} 
 and \eqref{eq:13}, respectively and $T_0\in \mathcal M_{\rm fin}([0,1])$. }

\begin{definition}
  \label{df012701-23}
{We say that
a function $T:[0,+\infty)\to \mathcal M_{\rm fin}([0,1])$ is a weak (measured
valued) solution of \eqref{eq:5} if:  it belongs to
$C\Big([0,+\infty); \mathcal M_{\rm fin}([0,1])\Big)$ and
  for any $\varphi\in C^2[0,1]$ such that 
     $\varphi(0)=\varphi'(1)=0$ we have
  \begin{equation}
    \label{eq:5w}
    \begin{split}
       \int_0^1\varphi(u)T(t,\dd u) - \int_0^1\varphi(u)T_0(\dd u) 
      =&\frac{D}{4\ga}\int_0^t\dd s\int_0^1  \varphi''(u)T(s,\dd u) \\
     &+\frac{DT_- t}{4\ga}\varphi'(0)-Jt \varphi(1).
    \end{split}
  \end{equation}}
\end{definition}

{
 The proof of the uniqueness of the solution of \eqref{eq:5w} is quite
 routine. For completeness sake we present it
   in \autoref{A:unique}.
  }

\begin{theorem}\label{th1}
  {Suppose that the initial configurations $(\mu_n)$ satisfy
    \eqref{eq:ass0}. Assume furthermore
    that there exists  $T_0\in  \mathcal M_+([0,1])$ such that
\begin{equation}
  \label{eq:6}
   \lim_{n\to\infty} \bbE_{\mu_n} \left[\frac 1{n+1} \sum_{x=0}^n \varphi\left(\frac x{n+1} \right)  \mathcal E_x(0)\right]
    = \int_0^1 \varphi (u) T_0( \dd u),
  \end{equation}
for any function $\varphi\in C[0,1]$ - the space of continuous
functions on $[0,1]$.
Here  
$\mathcal E_x(t) = \mc E_x({\bf q}(t),{\bf p}(t))$.}
{Then,
  \begin{equation}
    \label{eq:3}
    \lim_{n\to\infty} \frac 1{n+1}
    \sum_{x=0}^n \varphi\left(\frac x{n+1} \right) \bbE_{\mu_n}
    \left({\cal E}_x(n^2 t)\right) 
    = \int_0^1 \varphi (u) T(t,  \dd u).
  \end{equation}
  Here  $T(t,  \dd u)$ is the unique weak solution of 
 \eqref{eq:5}, with the initial data given by measure $T_0$ in \eqref{eq:6}.}
\end{theorem}

  \begin{remark}
 {   The initial energy $\cal E_x(0)$ can be represented
    as the sum $\cal E_x^{\rm th}+\cal E_x^{\rm mech}$ of the \emph{thermal energy}
    $$
   \cal E_x^{\rm th}:= \frac 12 \left[(p'_x)^2 + (q'_x - q'_{x-1})^2 + \omega_0^2
      (q'_x)^2\right]
    $$
     and the \emph{mechanical energy}
     $$
     \cal E_x^{\rm mech}:=\frac 12 \left[\bar p_x^2 + (\bar q_x - \bar q_{x-1})^2
       + \omega_0^2 \bar q_x^2\right].
     $$
    Here
    $q'_x = q_x - \bar q_x$ and $p'_x = p_x - \bar
    p_x$, with
$\bar p_x  :=  \int_{\Om_n}p_x
\mu_n(\dd\qv,\dd\pv)$ and $\bar q_x  :=  \int_{\Om_n}q_x
\mu_n(\dd\qv,\dd\pv)$.}
   
{   If $\cal E_x^{\rm mech}\not=0$ and
   satisfies \eqref{fHn}, with $\delta=1$, then 
    the initial measure $T_0(du)$ is the macroscopic distribution of the total energy and
    not of the  \emph{temperature}, where the latter is understood as the thermal energy.
    Nevertheless, as a consequence of Proposition \ref{prop011012-21},
    at any macroscopic positive time the entire mechanical energy is transformed
    immediately into the
    thermal energy, so that $T(t,  \dd u)$ for $t>0$ can be seen as the macroscopic
    temperature distribution. The situation is different for the unpinned dynamics
    ($\omega_0 = 0$) where the transfer of mechanical energy
    to thermal energy happens slowly at macroscopic times (see \cite{kos3}).}
  \end{remark}

  \begin{remark}
    \label{rem-mixed}
   { Concerning Theorem \ref{th1},  a similar proof will work}
    in the case where two Langevin heat baths
    at two temperatures, $T_-$ and $T_+$
    are  {placed} at the boundaries, in the absence of the periodic forcing.
    In this case the macroscopic equation
    will be the same but with boundary conditions
    $T(t,0) = T_-$ and $T (t,1) =T_+$.

    Also, in the  absence of any heat bath, we could apply two periodic forces $\mathcal F_n^{(0)}(t)$ and
    $\mathcal F_n^{(1)}(t)$ respectively at the left and right boundary.
    They will generate two incoming energy current, $J^{(0)}>0$ on the left and $J^{(1)}<0$ on the right,
    given by the corresponding formula \eqref{051021-05z}, and we will have
    the same equation but with boundary conditions
    $\partial_u T (t,0) = -\frac{4\gamma   J^{(0)}}{D}$ and $\partial_u T (t,1) = -\frac{4\gamma   J^{(1)}}{D}$.
    Of course in this case the total  energy increases in time and
    {periodic} stationary states do not exist.

    In the case where both a heat bath and a periodic force are present on the
    same side,
    say on the right endpoint,
    then the macroscopic boundary condition arising
    is  $T (t,1) =T_+$, i.e. the periodic forcing is
    ineffective on the macroscopic level,
    and all the energy generated by its work
    will flow into the heat bath. {It would be interesting to
    investigate what happens when the amplitude of the forcing is
    larger than considered here ($-1/2<a\le0$ in
    \eqref{Fnta-b}). However, it
    is not yet clear to us what occurs in this case.}
  \end{remark}

\begin{remark}
 { If the initial data $T_0$ is $C^1$ smooth and satisfies the boundary condition in \eqref{eq:5},
  then {the initial-boundary value  problem} \eqref{eq:5} has a unique strong solution $T(t,u)$ that belongs
  to the intersection of the spaces
  $C\big([0,+\infty)\times[0,1]\big)$ and $
  C^{1,2}\big((0,+\infty)\times(0,1)\big)$ - the space of functions
  continuously differentiable once in the first and twice in the
  second variable, see e.g. \cite[Corollary 5.3.2,
  p.147]{friedman}. This solution coincides then with the unique weak
  solution  in the sense of Definition \ref{df012701-23}.}
\end{remark}
\begin{remark}
  In the proof of \autoref{th1} we need to show a result about the equipartition of energy
  (cf. Proposition \ref{cor012912-21}). As a consequence the limit profile of the energy
  equals the limit profile of the temperature, i.e. we have
  {\begin{equation}
    \label{eq:3T}
   \lim_{n\to\infty} \frac 1{n+1}
    \sum_{x=0}^n \int_0^{+\infty}\varphi\left(t,\frac x{n+1} \right)  \bbE_{\mu_n}
    \left({p}_x^{2}(n^2 t)\right) \dd t
    = \int_0^{+\infty}\dd t\int_0^1 \varphi (t,u) T(t, \dd u),
  \end{equation}
  for any test function as in the statement of Theorem \ref{th1}.}
\end{remark}



\section{Entropy, energy and currents bounds}
\label{sec:bound-energ}

We first prove that the initial entropy bound \eqref{eq:ass0} holds for all times.
\begin{proposition}
\label{cor021211-19}
Suppose that the law of the initial configuration admits the density $ f_n(0,
{\bf q}, {\bf p})$ w.r.t. the Gibbs measure $\nu_{T-}$ that satisfies
\eqref{eq:ass0}. Then, for any $t>0$ there exists $ f_n(t, {\bf q},
{\bf p})$ - the density of the law of  the configuration
$\big({\bf q}(t), {\bf p}(t)\big)$.
 In addition, 
for any $t$  there exists a constant $C$ independent of $n$ such that 
\begin{equation}
  \label{eq:10bb}
 \sup_{s\in[0,t]} \mathbf{H}_{n}(f_n(n^2 s)) \le C n, 
\end{equation}
\end{proposition}
\proof
{For simplicity sake, we present  here a proof under an additional
  assumption that $ f_n(t,{\bf q}, {\bf p})$ is a smooth function such
  that $\mu_t(\dd\qv,\dd\pv)= f_n(t,{\bf q}, {\bf p}) \dd {\bf q}\dd
  {\bf p}$ is the solution  of the forward equation \eqref{eq:fwe}.
The general case is treated in Appendix \ref{secF}.}
Using \eqref{eq:7} for the generator ${\cal G}_t$ we conclude that
\begin{align*}
 \mathbf{H}_{n}(f_n(n^2 t)) - \mathbf{H}_{n}(f_n(0)) 
   =\int_0^{n^2 t} \dd s \int_{\Om_n} f_n(s) {\cal G}_{s} \log f_n(s) \dd \nu_{T-}
    ={\rm I}_n+{\rm II}_n,
\end{align*}
with
\begin{align*}
  &
    {\rm I}_n:=\ga \int_0^{n^2 t} \dd s \int_{\Om_n} f_n(s)\left({\cal S}_{\rm flip}
    + 2{\cal S}_{-}\right) \log f_n(s) \dd \nu_{T-},\\
  &
    {\rm II}_n :=\int_0^{n^2 t} \dd s \int_{\Om_n} f_n(s) {\cal A}_{ s}\log f_n(s) \dd \nu_{T-}.
\end{align*}
We have that $ {\rm I}_n\le 0$ because ${\cal S}_{\rm flip}$ and
${\cal S}_{-}$ are symmetric negative operators with respect to the measure $\nu_{T-}$.
The only positive contribution comes from the second term
where the boundary work defined by \eqref{eq:16} appears:
\begin{align*}
    {\rm II}_n = \int_0^{n^2 t} \dd s  \cF_n( s) \int_{\Om_n} \frac{p_n}{T_-} f_n(s) \dd  \nu_{T-}
       =-\frac{n}{T_-} {J_n(t,\mu_0)},
\end{align*}
{where $\dd\mu_0:=f_n(0) \dd  \nu_{T-}$}.
Therefore
\begin{align*}
  &\mathbf{H}_{n}(f_n(n^2 t))\le \mathbf{H}_{n}(f_n(0))
   -\frac{n }{T_-} J_n(t,\mu_0).
\end{align*}
The conclusion of the proposition then follows from a direct
application of \eqref{eq:ass0} and Theorem \ref{th4}. 
  \qed

 {To abbreviate the notation we shall omit the index by the
   expectation sign, indicating the initial
condition. }
 \begin{corollary}[Energy bound]\label{th2}
 For any $t_* \ge0$  we have
 \begin{equation}
  \label{eq:4a}
  \sup_{t\in[0,t_*]}\sup_{n\ge 1}
  \bbE \left[ \frac 1{n+1} \sum_{x=0}^n {\cal E}_x(n^2 t)\right] =E(t_*)<+\infty.
\end{equation}
\end{corollary}

\begin{proof}
It follows from the entropy inequality, see e.g.~
  \cite[p.~338]{KL}, that for $\al>0$ small enough we can find $C_\al>0$ such that
\begin{equation}\label{eq:boundent}
 \bbE \left[\sum_{x=0}^n{\cal E}_x(n^2 t)\right]
  \le \frac{1}{\al}\big(C_\al n+ \mathbf{H}_{n}(t)\big),\qquad t\ge0 .
 \end{equation}
\end{proof}

From Theorem \ref{th4} and  Corollary \ref{th2}  we
immediately conclude the following.
\begin{corollary}[Current size]\label{lem:bound2}
  For any $t_* \ge0$ there exists $C>0$ such that
  {\begin{equation}
    \label{eq:11}
           \sup_{x=0,\ldots,n+1,\,t\in[0,t_*]}\left|\int_0^t
             \mathbb{E} \left[  j_{x-1,x}(n^2s)\right]
             \dd s\right| \le  \frac{C}{ n },\quad n=1,2,\ldots.
  \end{equation}}
  In particular, for any $t>0$ there exists $C>0$ such that 
 \begin{equation}
    \label{eq:32c}
     \Big|
       \int_0^{t}  \Big\{\mathbb
    E \big[ p_0^2(n^2s)   \big]- T_-  
    \Big\}\dd s\Big|\le\frac{C}{n},
  \end{equation}
\end{corollary}
\proof
By the local conservation of energy
\begin{equation}
\label{eq:current-n1}
n^{-2}\frac{\dd}{\dd t}\bbE[{\cal E}_x(n^2 t)]=
\bbE\Big[j_{x-1,x}(n^2 t) -j_{x,x+1}(n^2 t)\Big] .
\end{equation}
Therefore
\begin{align}
  \label{eq:current-2}
  \int_0^t \bbE j_{x-1,x}(n^2 s)\dd s=\int_0^t \bbE j_{n,n+1}(n^2 s)\dd s
  +n^{-2}\sum_{y=x}^n\Big(\bbE[{\cal E}_y(n^2 t)]
  -\bbE[{\cal E}_y(0)]\Big),
\end{align}
and bound \eqref{eq:11} follows directly from estimates \eqref{Jnn}
and  \eqref{eq:4a}. Estimate \eqref{eq:32c} is a consequence of the
definition of $j_{-1,0}$ (see \eqref{eq:current-bound}) and \eqref{eq:11}.
\qed

\section{Equipartition of energy and
  Fluctuation-dissipation relations}
\label{sec:equip-energy-fluct}

       \subsection{Equipartition of the energy}

In the present section we show the equipartition
property of the energy.
\begin{proposition}
\label{cor012912-21}
  Suppose that $\varphi\in {C^1}[0,1]$ is such that ${\rm supp}\,\varphi\subset(0,1)$. Then,
  \begin{equation}
    \label{011312-21}
\lim_{n\to+\infty}\frac{1}{n+1}\sum_{x=0}^n
\varphi\left(\frac{x}{n+1}\right) \int_0^{t} \mathbb E\Big[
p_x^2(n^2s)- \big(q_x(n^2s)-q_{x-1}(n^2s)\big)^2-\om_0^2  q_x^2(n^2s)\Big]\dd s=0.
\end{equation}
  \end{proposition}

  \proof
  {After a simple calculation we obtain the following
  fluctuation-dissipation relation: for $x= 1, \dots, n-1$, }
  \begin{equation}
    \label{eq:4}
    p_x^2 - \omega_0^2 q_x^2 - (q_x - q_{x-1})^2 = \nabla^\star\left[q_x(q_{x+1} - q_x)\right]
    + \mathcal G_t\left(q_x p_x + \gamma q_x^2\right),
  \end{equation}
  {where the discrete gradient $\nabla$ and its adjoint $\nabla^\star$ are defined in \eqref{012901-23} below.}

Therefore,
\begin{align}
  \label{010905-22}
  & \int_0^t \mathbb E \Big[p_x^2(n^2s) - \om^2_0q_x^2(n^2s)
    -(q_x(n^2s) - q_x(n^2 s))^2\Big] \dd s\\
  &
    =
    \nabla \int_0^t \mathbb E\left[q_x(n^2s)(q_{x+1}(n^2s) - q_x(n^2s))\right] \dd s
    \notag  \\
   & +n^{-2} \mathbb E
     \Big[q_x(n^2t)p_x(n^2t) +2\ga q_x^2(n^2t) \Big]
     - n^{-2} \mathbb E
    \Big[q_x(0)p_x(0) +2\ga q_x^2(0) \Big].\notag 
\end{align}
After summing up against the test function $\varphi$ (that has compact support strictly contained in $(0,1)$) and using the energy bound
\eqref{eq:4a}
we conclude \eqref{011312-21}.

\qed


\subsection{Fluctuation-dissipation relation}

In analogy to  \cite[Section 5.1]{klo22} define
\begin{equation}
  \label{eq:32}
  \begin{split}
  \frak f_x &:= \frac 1{4\gamma} \left(q_{x+1} - q_x\right)\left(p_x + p_{x+1}\right)
  + \frac 14 \left(q_{x+1} - q_x\right)^2 ,\qquad x=0, \dots, n-1,\\
\mathfrak F_x &:= p_x^2 + \left(q_{x+1} - q_x\right) \left(q_{x} - q_{x-1}\right)
    -\omega_0^2 q_x^2,\qquad x=0,\dots, n,
\end{split}
\end{equation}
with the convention that $q_{-1} = q_{0}$, $q_n=q_{n+1}$.
Then
\begin{equation}
  \label{eq:33}
  j_{x,x+1} =  - \frac 1{4\gamma} \nabla  \mathfrak F_x + \mathcal G_t \frak f_x
    - \frac {\delta_{x,n-1}} {4\gamma} \cF_n(t) \left(q_{n} - q_{n-1}\right),
    \quad x=0, \dots, n-1.
  \end{equation}

\section{Local equilibrium and the Proof of Theorem \ref{th1}}
\label{sec:proofth1}

The fundamental ingredients in the proof of Theorem \ref{th1} are
the identification of the work done at the boundary given by \autoref{th4},
the equipartition and the fluctuation-dissipation relation contained in \autoref{sec:equip-energy-fluct}, and
 the following \emph{local equilibrium} results. In the bulk we have the following:
\begin{proposition}
\label{prop010803-22}
Suppose that $\varphi\in C[0,1]$ is such that ${\rm
  supp}\,\varphi\subset(0,1)$. Then 
  \begin{equation}
    \label{eq:32a}
    \lim_{n\to+\infty}\frac{1}{n+1} \sum_{x=0}^{n}
    \varphi\left(\frac{x}{n+1}\right) \int_0^{t}  \mathbb E
    \big[ q_x(n^2s) q_{x+\ell}(n^2s) -
    G_{\omega_0}(\ell) p_x^2(n^2s) \big] \dd s=0,
  \end{equation}
 for $\ell=0,1,2$.  {Here $G_{\omega_0}(\ell)$ is the
Green's function of $-\Delta_{\bbZ} + \omega^2_0$, where
$\Delta_{\bbZ}$ is   the lattice laplacian, see \eqref{GR}.}
\end{proposition}
{At the left boundary the situation is a bit different, due to the
  fact that $q_0=q_{-1}$, and we have}
\begin{proposition}
  \label{prop-boundaryeq}
  We have
  \begin{equation}
    \label{eq:20}
    \lim_{n\to+\infty} 
       \int_0^{t}  \mathbb E\big[ q_0^2(n^2 s) - \left(G_{\omega_0}(1) +
         G_{\omega_0}(0)\right)  p_0^2(n^2s) \big]  \dd s =0.
  \end{equation}
\end{proposition}

The proofs of Propositions \ref{prop010803-22} and \ref{prop-boundaryeq}
require the analysis of the evolution
of the covariance matrix of {the position and momenta vector}
and will be done in \autoref{sec-proofs-le}.
As a consequence, recalling definition \eqref{eq:32},  the bound
\eqref{eq:32c} and {the identity $2G_{\om_0} (1)-G_{\om_0} (0)
        -G_{\om_0} (2) =-\om_0^2G_{\om_0} (1)$}
we have the following corollary
\begin{corollary}
\label{cor-loceq}
For any $t>0$ and $\varphi\in C[0,1]$ such that ${\rm
  supp}\,\varphi\subset(0,1)$ we have
\begin{equation}
    \label{eq:32d}
    \lim_{n\to+\infty}\frac{1}{n+1} \sum_{x=0}^{n}
    \varphi\left(\frac{x}{n+1}\right) \int_0^{t}  \mathbb E
    \big[ \mathfrak F_x(n^2s) - D p_x^2(n^2s) \big] \dd s=0
  \end{equation}
  and
\begin{equation}
    \label{eq:32cc}
    \lim_{n\to+\infty} 
       \int_0^{t}  \Big\{\mathbb
    E\big[\mathfrak F_0(n^2s)   \big]- DT_-  
    \Big\}\dd s =0.
  \end{equation}
 { Here $D$ is defined in \eqref{eq:13}.}
\end{corollary}

\subsection{Proof of Theorem \ref{th1}}

{
  Consider the subset $\mathcal M_{+,E_*}([0,1])$ of  $\mathcal
   M_+([0,1])$ (the space of all positive, finite
  {Borel} measures on  $[0,1]$)
consisting of measures with total mass less than or equal to $E_*$. It is
compact in the topology of weak convergence of measures. In addition,  the topology is metrizable
when restricted to this set.

For any $t\in [0,t_*]$ and $\varphi \in  C[0,1]$ define
\begin{equation}
  \label{eq:9}
  \begin{split}
  \xi_n(t, \varphi) =
  \frac 1{n+1} \sum_{x=0}^n \varphi_x \bbE \big[{\cal E}_x(n^2t)\big],\qquad
  \varphi_x := \varphi\left(\frac {x}{n+1}\right)
\end{split}
\end{equation}
{for any $\varphi\in C[0,1]$.}
Since  flips  of the momenta do not affect the energies ${\cal E}_x$, we have 
$\xi_n \in C\left([0,t_*], \mathcal M_+([0,1])\right)$.
Here  $C\left([0,t_*], \mathcal M_{+,E_*}([0,1])\right)$
is endowed with the   topology of the uniform convergence.
As a consequence of Corollary \ref{th2} for any $t_*>0$ the total energy is bounded by
$E_* = E(t_*)$ (see \eqref{eq:4a}) and we have that
$\xi_n \in C\left([0,t_*], \mathcal M_{+,E_*}([0,1])\right)$.

\subsection{Compactness}
\label{ssec:compact}

  Since $\mathcal M_{+,E_*}([0,1])$ is compact, in order to show that $(\xi_n)$
  is compact, 
  we only need to control  
modulus of continuity in time of  $\xi_n(t, \varphi)$ for any
$\varphi\in C^1[0,1]$, see e.g. ~\cite[p. 234]{kelley}.
This will be consequence of the following Proposition.

\begin{proposition}
  \label{mod-cont}
  \begin{equation}
    \label{eq:14}
    \lim_{\delta \downarrow 0} \limsup_{n\to\infty} 
      \sup_{0\le s ,t\le t_*, |t-s| < \delta} \left| \xi_n(t, \varphi) - \xi_n(s, \varphi)\right| = 0
  \end{equation}
\end{proposition}
{The proof of Proposition \ref{mod-cont} is postponed untill Section
\ref{cc}, we first use it to proceed  with 
 the limit identification argument. }

\subsection{Limit identification}

Consider a smooth test function $\varphi\in C^2[0,1]$
such that
\begin{equation}
  \label{010805-22}
  \varphi(0)= \varphi'(1)=0.
\end{equation}
{In what follows  we use the following notation. For a given function $\varphi:[0,1]\to\bbR$ and $n=1,2,\ldots$ we define  discrete  approximations of the function itself and of its gradient, respectively by
 \begin{align}
\label{disc-approx}
\varphi_x:=\varphi (\tfrac x {n+1}) ,\quad(\nabla_n \varphi)_x:=(n+1)\big(\varphi (\tfrac{x+1}{n+1})-\varphi
                      (\tfrac x {n+1})\big),
                                                \,\mbox{ for }x\in\{0,\ldots,n\}. 
\end{align}  
We use the convention  
$\varphi (-\tfrac{1}{n+1})=\varphi (0)$.}
Let $0<t_* <+\infty$ be fixed.
In what follows we show that, for any $t\in [0,t_*]$
 \begin{equation}
\label{012405-22}
\begin{split}
\xi_n(t,\varphi) - \xi_n(0,\varphi) =  \frac{\varphi'(0)DT_-t}{4\ga}
-Jt\varphi(1) 
+\frac{D}{4\ga}\int_0^{t} \xi_n(s,\varphi'') \dd s + o_n.
\end{split}
\end{equation}
Here, and in what follows $o_n$ denotes a quantity satisfying $\lim_{n\to+\infty}o_n=0$.
Thus any limiting point of $\big(\xi_n(t)\big)$
has to be the unique weak solution of \eqref{eq:5w} and this obviously
proves the conclusion of Theorem \ref{th1}. 

By  an approximation argument we can restrict ourselves to the case
when ${\rm supp}\,\varphi''\subset(0,1)$.  
  Then as in \eqref{eq:12} we have
\begin{equation}
  \label{eq:12b}
  \begin{split}
  \xi_n(t, \varphi) -  \xi_n(0, \varphi)   
  = \frac{n^2}{n+1} \sum_{x=0}^{n-1} \left(\varphi_{x+1} - \varphi_x\right)
  \int_0^{t} \mathbb E\left[ j_{x,x+1} (n^2\tau)\right] \dd \tau\\
  - \frac{n^2}{n+1}  \varphi_n \int_0^{t}\mathbb E\left[ j_{n,n+1} (n^2\tau)\right]  \dd \tau,
\end{split}
\end{equation}
By Theorem \ref{th4} the last term converges to $- \varphi(1) J t$.
        On the other hand from \eqref{eq:33}  we have
\begin{equation}
  \label{eq:15}
   \frac{n^2}{n+1} \sum_{x=0}^{n-1} \left(\varphi_{x+1} - \varphi_x\right)
   \int_0^{t} \mathbb E\left[ j_{x,x+1} (n^2\tau)\right] \dd \tau
    =\sum_{j=1}^3\mathrm{I}_{n,j},
  \end{equation}
  where
 \begin{align*}
   &\mathrm{I}_{n,1}:=-\frac{1}{4\ga}\left(\frac{n}{n+1}\right)^2
     \sum_{x=0}^{n-1} \nabla_n \varphi_x  \int_0^{t}
     \mathbb E\big[ \; \nabla \mathfrak F_x(n^2s)
     \big]  \dd s,\\
          &\mathrm{I}_{n,2}
   :=\left(\frac{1}{n+1}\right)^2\sum_{x=0}^{n-1} \nabla_n \varphi_x
     \bbE\Big[ \frak f_x(n^2t)- \frak f_x(0)\Big]  \dd s,
   \\
   &\mathrm{I}_{n,3}
   :=-\frac{1}{4\ga}\left(\frac{n}{n+1}\right)^2 
     \nabla_n \varphi_{n-1}  \int_0^{t}\cF_n(n^2s) \mathbb E\big[ q_{n}(n^2s) - q_{n-1}(n^2s)\big]\dd s.
 \end{align*}
 
 It is easy to see from Corollary \ref{th2} that
 $ \mathrm{I}_{n,2} =\bar o_n(t).$
Here the symbol $\bar o_n(t)$ stands for a quantity that satisfies
\begin{equation}
  \label{040506-22}
 \lim_{n\to+\infty}\sup_{s\in[0,t_*]}|\bar o_n(s)|=0.
  \end{equation}
 
Using the fact that $\varphi'(1)=0$ and the {estimate \eqref{021012-21}}
respectively we conclude also that $\mathrm{I}_{n,3} =\bar o_n(t)$.
Thanks to Corollary \ref{th2} and \eqref{010805-22} we have
\begin{align*}
   &\mathrm{I}_{n,1}= \sum_{j=1}^3\mathrm{I}_{n,1}^{(j)} + \bar o_n(t),\quad
     \mbox{where}\\
   &
     \mathrm{I}_{n,1}^{(1)}:=
     \frac{1}{4\ga (n+1)} \sum_{x=0}^{n}  \varphi''\left(\frac{x}{n+1}\right) \int_0^{t}
     \mathbb E\big[
       \;  
     \mathfrak F_x(n^2s)  \big]  \dd s,\\
   &
     \mathrm{I}_{n,1}^{(2)}:=-\frac{1}{4\ga}\left(\frac{n}{n+1}\right)^2  \varphi'\left(\frac{n-1}{n+1}\right)
     \int_0^{t} \mathbb E\big[\; \mathfrak F_n(n^2s)  \big]  \dd s =
     \bar o_n(t),\\
   &
    \mathrm{I}_{n,1}^{(3)}:= \frac{1}{4\ga} 
    \varphi'(0) \int_0^{t} \mathbb E\big[
       \;  
       \mathfrak F_0(n^2s)  \big]  \dd s.
 \end{align*}

 {Since ${\rm supp}\,\varphi''\subset(0,1)$, by Corollary \ref{cor-loceq} 
 and {the equipartition property}  (Proposition \ref{cor012912-21})
 for a fixed $t\in[0,t_*]$ we have
 \begin{align}
   \label{012402-23}
     \mathrm{I}_{n,1}^{(1)}
    & =\frac{D}{4\ga (n+1)} \sum_{x=0}^{n}\varphi''
     \left(\frac{x}{n+1}\right) \int_0^{t} \mathbb E\Big[ \;  
       p_x^2(n^2s)  \Big]  \dd s+\bar o_n(t)\\
  & =\frac{D}{4\ga (n+1)} \sum_{x=0}^{n}\varphi''
     \left(\frac{x}{n+1}\right) \int_0^{t} \mathbb E\Big[ \;  
       \mathcal E_x(n^2s)  \Big]  \dd s+o_n.\notag
\end{align}
Concluding, we have obtained
 \begin{align*}
    \lim_{n\to+\infty}\mathrm{I}_{n,1}^{(3)} = \frac{\varphi'(0)DT_-t}{4\ga}.
 \end{align*}}
Thus \eqref{012405-22} follows.\qed

\subsection{Proof of Proposition \ref{mod-cont}}
\label{cc}

{From the calculation made in \eqref{eq:12b}--\eqref{012402-23} we
  conclude that for any function
  $\varphi\in C^2[0,1]$ satisfying \eqref{010805-22} we have
\begin{equation}
  \label{eq:12}
  \begin{split}
    &\xi_n(t, \varphi) -  \xi_n(s, \varphi) =
     \frac{\varphi'(0)DT_-}{4\ga} (t-s)
-J\varphi(1) (t-s)
\\
&
+\frac{D}{4\ga (n+1)} \sum_{x=0}^{n}\varphi''
     \left(\frac{x}{n+1}\right) \int_s^{t} \mathbb E\Big[ \;  
       p_x^2(n^2\tau )  \Big]  \dd \tau + \bar o_n(t) + \bar o_n(s)
\end{split}
\end{equation}
for any $ 0\le s<t \le t_*$
and \eqref{eq:14} follows immediately.
A density argument completes the proof.}

}

\section*{Acknowledgements} {The authors would like to thank both 
anonymous referees for their careful reading of the manuscript and
many helpful remarks that lead to the improvement of the paper.}

\section*{Funding, Data Availability and Conflict of Interest
  Statements}
 The work of J.L.L. was supported in
  part by the A.F.O.S.R. He thanks the Institute for Advanced Studies
  for its hospitality.
  T.K. acknowledges the support of the NCN grant 2020/37/B/ST1/00426. 

{\bf Data sharing is not applicable to this article} as no datasets were
generated or analysed during the current study. In addition, the
authors {\bf have no conflicts of interest to declare} that are relevant to the content of this article.

\appendix

\section{The discrete laplacian}
\label{sec:discrete-laplacian}

\subsection{Discrete gradient and laplacian}

\label{sec2.7}

Recall that the lattice gradient, its adjoint and laplacian of any
$f:\bbZ\to\bbR$ are  defined as
{\begin{equation}
  \label{012901-23}
  \nabla
f_x=f_{x+1}- f_{x}, \qquad \nabla^\star
f_x=f_{x-1}- f_{x}
\end{equation}
and $\Delta_\bbZ f_x=-\nabla^\star \nabla f_x=f_{x+1}+
f_{x-1}-2 f_{x}$, $x\in\bbZ$, respectively.}

Suppose that $\om_0>0$. Consider the
Green's function of $-\Delta_{\bbZ} + \omega^2_0$, where  $\Delta_{\bbZ}$ is the laplacian on the integer lattice $\bbZ$. It is given by, see
e.g. \cite[(27)]{ray},
\begin{align}
\label{GR}
&G_{\om_0}(x) = \left(-\Delta_{\bbZ} + \omega^2_0 \right)^{-1}(x)\quad 
=\int_0^1\left\{4\sin^2(\pi u)+\om_0^2\right\}^{-1}\cos(2\pi ux)du\\
&
 =\frac{1}{\om_0\sqrt{\om_0^2+4}}\left\{1+\frac{\om_0^2}{2}+\om_0\sqrt{1
    +\frac{\om_0^2}{4}}\right\}^{-|x|},\quad x\in\bbZ.\notag
\end{align}

\subsection{Discrete Neumann laplacian $-\Delta$}

{Let  $\la_j$ and $\psi_j$, $j=0,\ldots,n$ be the respective  eigenvalues
and eigenfunctions for the discrete Neumann laplacian $-\Delta$
defined in \eqref{NL}.}
They are given by
\begin{align}
\label{laps}
\la_j=4\sin^2\left(\frac{\pi j}{2(n+1)}\right),\quad
\psi_j(x)=\left(\frac{2-\delta_{0,j}}{n+1}\right)^{1/2}\cos\left(\frac{\pi
    j(2x+1)}{2(n+1)}\right),\quad x,j=0,\ldots,n.
\end{align}
The eigenvalues of $\om_0^2-\Delta$ are
given by
\begin{equation}
  \label{muj}
\mu_j=\om_0^2+\la_j=\om_0^2+4\sin^2\left(\frac{\pi
    j}{2(n+1)}\right),\quad j=0,\ldots,n.
\end{equation}

  \section{The dynamics of the  means}

  \label{app-means}

Let $\mu$ be a Borel probability measure on $\bbR^{2(n+1)}$ and let
$(\bar{\qv}, \bar{\pv})$ be the vector of the $\mu$-averages of initial
data. In the following we denote by
$ \left(\begin{matrix}\bar{\qv}(t)\\
    \bar{\pv}(t)
 \end{matrix}\right)$ the vector means of positions and momenta by
$\bar q_x(t) = \bbE_{{\bf q},{\bf p}}(q_x(t))$ and
$\bar p_x(t) = \bbE_{{\bf q},{\bf p}}(p_x(t))$.
 Let ${\bf e}_{2(n+1)}$ be the $2(n+1)$ vector defined by ${\rm e}_{2(n+1),j}= \delta_{2(n+1),j}$.
{Then, performing the averages in \eqref{eq:flip} and \eqref{eq:pbdf},
we conclude  the evolution equation for the averages. Its solution is given by}
\begin{equation}
  \begin{split}
  \label{011605-22}
\left(
  \begin{array}{c}\bar{\bf q}(t)\\
    \bar{\bf p}(t)
    \end{array}\right)=e^{-At}\left(
  \begin{array}{c}\bar{\bf q}\\
    \bar{\bf p}
    \end{array}\right)+  \int_0^t \cF_n\left(s\right)
  e^{-A(t-s)}\; {\bf e}_{2(n+1)} \dd s.
\end{split}
\end{equation}
Here $A$ is a $2\times 2$ block matrix made of $(n+1)\times (n+1)$
matrices of the form
\begin{equation}
\label{A}
A=
\left(
  \begin{array}{cc}
    0&-{\rm Id}_{n+1}\\
    -\Delta +\om_0^2 {\rm Id}_{n+1}& 2\ga {\rm Id}_{n+1}
  \end{array}
\right),
\end{equation}
where ${\rm Id}_{n+1}$ is  the $(n+1)\times (n+1)$ identity matrix.

Using the expansion
$$
{\cal F}(t)=\sum_{\ell\in\bbZ}\hat {\cal F}(\ell)e^{2\pi i \ell t}
$$
and defining
\begin{align}
  \label{eq:yz}
& \left(
  \begin{array}{c}\bar{\bf y}(t)\\
    \bar{\bf z}(t)
  \end{array}\right) := e^{-A t}{\bf e}_{2(n+1)},
  \end{align}
we can write
\begin{equation}
    \label{021605-22}
    \begin{split}
    \left(
  \begin{array}{c}\bar{\bf q}(t)\\
    \bar{\bf p}(t)
  \end{array}\right)
=&e^{-At}\left(
  \begin{array}{c}\bar{\bf q}\\
    \bar{\bf p}
  \end{array}\right) +
\sum_{\ell\in\bbZ} \frac {\hat{\cal F}(\ell) }{\sqrt n} \int_0^{t} e^{2\pi i \ell s/\theta}
 \left(
  \begin{array}{c}\bar{\bf y}(t-s)\\
    \bar{\bf z}(t-s)
  \end{array}
\right)\dd s.
\end{split}
\end{equation}

To find the formulas for the components of  $\bar  v_x(t)$, $\bar
u_x(t)$, $x=0,\ldots,n$ of  the vector $\left(
  \begin{array}{c}\bar{\bf u}(t)\\
    \bar{\bf v}(t)
  \end{array}\right) : =e^{-At}\left(
  \begin{array}{c}\bar{\bf q}\\
    \bar{\bf p}
  \end{array}\right) $    it is convenient to use the Fourier
coordinates in the base $\psi_j$ of the eigenvectors for the Neumann laplacian
$\Delta$, see \eqref{laps}. Let
$\tilde u_j(t) =\sum_{x=0}^n\bar  u_x(t)\psi_j(x)$ and
$\tilde v_j(t) =\sum_{x=0}^n\bar  v_x(t)\psi_j(x)$ be  the Fourier
coordinates  of the vector $\left(
   \bar{\bf u}(t),
    \bar{\bf v}(t)
   \right)$.  Likewise, we let
$\tilde q_j =\sum_{x=0}^n \bar q_x\psi_j(x)$ and
$\tilde p_j =\sum_{x=0}^n  \bar p_x\psi_j(x)$, with
 $\bar  q_x$, $\bar
p_x$, $x=0,\ldots,n$ the components of  $(\bar{\bf q}, \bar{\bf p})$.

 Let 
\begin{equation}
  \label{lapm}
\la_{j,\pm}:=\ga\pm\sqrt{\ga^2-\mu_j}
\end{equation}
be the  two solutions 
of the equation
\begin{equation}
  \label{011602-23}
\la^2-2\ga \la+\mu_j=0.
\end{equation}
Note that
$
\la_{j,+}\la_{j,-}=\mu_j.
$
Then,
\begin{align}
  \label{012810-21y}
    \tilde
    u_j (t) 
  = \frac{1}{2\sqrt{\ga^2-\mu_j}}\Big[
   - (\la_{j,-}\tilde q_j+\tilde p_j)\exp\left\{-\la_{j,+}t\right\}+
   (\la_{j,+}\tilde q_j+\tilde p_j) \exp\left\{-\la_{j,-}t\right\}\Big].
\end{align}
and
\begin{align}
\label{012810-21}
    \tilde
    v_j (t)
    = \frac{1}{2\sqrt{\ga^2-\mu_j}}\Big[
   (\mu_{j}\tilde q_j+\la_{j,+}\tilde p_j)\exp\left\{-\la_{j,+}t\right\}-
   ( \mu_{j}\tilde q_j+\la_{j,-}\tilde p_j)
    \exp\left\{-\la_{j,-}t\right\}\Big],
\end{align}
in the case when $\mu_j\not=\ga^2$. When $\ga^2=\mu_j$ (then $\la_{j,\pm}=\ga$) we have
 $$
 \tilde   u_j (t)= \Big[(1+\ga t) \tilde   q_j +\tilde   p_j t\Big]e^{-\ga t},\qquad \tilde   v_j (t)=
   \Big[\ga^2t  \tilde   q_j + (1-\ga t)\tilde   p_j \Big]e^{-\ga t},
   $$
Then, by \eqref{012810-21y} and \eqref{012810-21}, we conclude that
the components of $e^{-A t}{\bf e}_{2(n+1)}$ equal
\begin{equation}
   \begin{split}
\label{eq:zyj}
&\tilde
    y_j (t)
    =\frac{\psi_j(n)}{2\sqrt{\ga^2-\mu_j}}\Big(
   - \exp\left\{-\la_{j,+}t\right\}+
  \exp\left\{-\la_{j,-}t\right\}\Big),\\
  &{\tilde z}_j (t)
      =
    \frac{\psi_j(n)}{2\sqrt{\ga^2-\mu_j}}\Big(
  \la_{j,+}\exp\left\{-\la_{j,+}t\right\}-
   \la_{j,-}
    \exp\left\{-\la_{j,-}t\right\}\Big).
     \end{split}
   \end{equation}
   in the case when $\mu_j\not=\ga^2$.
In the case that $\ga^2=\mu_j$ (then $\la_{j,\pm}=\ga$) we have
$\tilde   y_j (t)= \psi_j(n) te^{-\ga t}$ and   $\tilde   z_j (t)=
\psi_j(n) (1-\ga t)e^{-\ga t}$.

 Elementary calculations lead to the following bounds
\begin{equation}
  \label{010606-22}
{\rm Re}\la_{j,\pm}\ge \ga_*:=\min\left\{\ga,
  \frac{\om_0^2}{2\ga}\right\},\quad |\la_{j,\pm}|\le \ga+|\ga^2-\om_0^2-4|^{1/2},\quad  j=0,\ldots,n.
\end{equation}
Hence, there exists $C>0$ such that
\begin{equation}
  \label{012008-22}
  |\tilde   y_j (t) |+|\tilde   z_j (t) |\le
  C(t+1)e^{-\ga_*t}|\psi_j(n)|
\end{equation}
for all $t\ge0$,  $j=0,\ldots,n$, $n=1,2,\ldots$.
By the Plancherel identity, \eqref{012810-21y} and \eqref{012810-21}
we conclude also that
there exist  constants $C,C'>0$ such that, for all $t\ge0$  and $n\in \mathbb N$,
\begin{equation}
  \label{011908-22}
  \begin{split}
    &\sum_{x=0}^n\Big[\bar  u^2_x (t)+\bar   v^2_x (t)\Big]=\sum_{j=0}^n\Big[\tilde
    u^2_j (t)+\tilde   v^2_j (t)\Big]\\
    &
    \le
C(t+1)e^{-\ga_*t}\sum_{j=0}^n\Big(\tilde   q^2_j +\tilde   p^2_j
\Big)\le C'(t+1)e^{-\ga_*t}{\cal H}_n\Big(\bar{\bf q}, \bar{\bf p}\Big).
\end{split}
\end{equation}

 \subsection{$L^2$ norms of the means}

 By \eqref{021605-22}, {the triangle inequality} and the Plancherel theorem  
    \begin{equation}
      \begin{split}
     &   \sum_{x=0}^n\int_0^t\left[\bar q_x^2(s)+\bar p_x^2(s)\right]\dd s
\le    C  \sum_{x=0}^n\int_0^t\left[\bar u_x^2(s)+\bar
  v_x^2(s)\right]\dd s \\
&
+\frac Cn \sum_{j=0}^n \left[\left| \sum_{\ell\in\bbZ} \hat \cF(\ell)\int_0^t
          e^{2\pi i\ell s/\theta} \tilde y_j(t-s) \dd s\right|^2 +\left| \sum_{\ell\in\bbZ}\hat \cF(\ell)\int_0^t
          e^{2\pi i\ell s/\theta} \tilde z_j(t-s) \dd s\right|^2\right] .
  \end{split}
\end{equation}
The constant appearing here below do not depend on $t$ and $n$. 
Using \eqref{eq:2}, \eqref{012008-22}
and \eqref{011908-22}
we conclude therefore that
\begin{equation}
  \label{012208-22}
      \begin{split}
     &   \sum_{x=0}^n\int_0^t\left[\bar q_x^2(s)+\bar p_x^2(s)\right]\dd s
\le    C{\cal H}_n\Big(\bar{\bf q}, \bar{\bf p}\Big) 
+\frac Cn \left( \sum_{\ell\in\bbZ}|\hat \cF(\ell)|\right)^2  .
  \end{split}
\end{equation}
From \eqref{012208-22} we conclude therefore the following.
  \begin{proposition}
\label{prop011012-21}
Assume that the hypotheses of Theorem \ref{th4} are in force. Then,   there exists $C>0$ such that
\begin{equation}
\label{021012-21}
 \sum_{x=0}^n \int_0^t\Big[\bar q_x^2(n^2s)+\bar p_x^2(n^2s)\Big]\dd s\le 
{\frac{C}{n^{\kappa}}}, 
\end{equation}
for all $t\ge0$, $n=1,2,\ldots$. Here $\kappa=\min\{2-\delta,1\}$ and
$\delta$ is as in \eqref{fHn}. {If  the hypotheses of Theorem
\ref{th1}  hold, then $\delta=1$ and \eqref{021012-21} is satisfied
with $\kappa=1$.}
\end{proposition}

\subsection{The proof of Theorem \ref{th4}}

\label{sec-th4}

We   show \eqref{Jnn} and \eqref{051021-05z}.
 Recall that the initial configuration $({\bf q},{\bf p})$ is
 distributed according to $\mu_n$. For the work done we have
\begin{equation}
  \label{eq:19}
  \begin{split}
    &  W_n(t) := \int_0^t \cF_n(s) \bar p_n(s) \dd s  \\
    &
    =\sum_j \psi_j(n) \sum_{\ell\in\bbZ} \frac 1{\sqrt n}
    \hat\cF(\ell)^\star\int_0^t e^{-i2\pi \ell s/\theta} \tilde p_j(s)\dd s .
  \end{split}
\end{equation}
We have $J_n(t;\mu)=-W_n(n^2t)/n$, see \eqref{eq:16}.

Using \eqref{011605-22} the utmost right hand side can be rewritten in the form
$W_{n,i}(t)+W_{n,f}(t) $ where
 \begin{equation}
   \label{eq:19a}
  \begin{split}
 &W_{n,i}(t):=\sum_{j=0}^n \psi_j(n) \sum_{\ell\in\bbZ} \frac 1{\sqrt n}
    \hat\cF(\ell)^*\int_0^t e^{-i2\pi \ell s/\theta} \tilde v_j(s) \dd s ,\\
   &W_{n,f}(t):= \frac 1n  \sum_{j=0}^n \psi_j(n) \sum_{\ell,\ell'\in\bbZ} 
    \hat\cF(\ell) ^* \hat\cF(\ell')
    \int_0^t \dd s\;  e^{i2\pi (\ell'-\ell) s/\theta} \int_0^s
    e^{-i2\pi \ell' s'/\theta} \tilde z_j(s') \dd s',
     \end{split}
   \end{equation}
   with $\tilde v_j(s)$ and $\tilde z_j(s')$ defined in
   \eqref{012810-21} and \eqref{eq:zyj}.

   Thanks to {the last estimate of \eqref{eq:2}} and the Cauchy-Schwarz inequality we
   conclude from \eqref{011908-22}
   $$
|W_{n,i}(t)|\le \frac C{\sqrt n}{\cal H}_n^{1/2}\Big(\bar{\bf q}_n, \bar{\bf p}_n\Big)\int_0^t (s+1)^{1/2}e^{-\ga_*s/2}\dd s.
$$
Thanks to \eqref{fHn} $\lim_{n\to+\infty}|W_{n,i}(n^2t)|/n=0$.
Using \eqref{eq:zyj} we have
\begin{equation*}
  \begin{split}
    \int_0^s e^{-i2\pi \ell' s'/\theta} \tilde z_j(s') \dd s'\\
    = \frac{\psi_j(n)}{\la_{j,+} - \la_{j,-}}
    \left[
       \frac{\la_{j,-} \left[e^{-s\left(\la_{j,-} +2\pi i \ell'  /\theta\right)}-1\right]}{2\pi i \ell'  /\theta+\la_{j,-}}
 -\frac{\la_{j,+} \left[e^{-s\left(\la_{j,+}+2\pi i \ell'  /\theta\right)}-1\right]}{2\pi i \ell'  /\theta+\la_{j,+}}   \right],
  \end{split}
\end{equation*}
so that we can decompose the work done in
$W_{n,f}(t) =W_{n,f}^{(1)} (t)+W_{n,f}^{(2)} (t),$
where
\begin{equation*}
  \begin{split}
 &\frac 1n W_{n,f}^{(1)}(n^2t) 
  := - \frac 1{n^2}  \sum_{j=0}^n\frac{\psi_j^2(n)}{\la_{j,+} - \la_{j,-}}
    \sum_{\ell,\ell'\in\bbZ} \hat {\cal F}^*(\ell) \hat {\cal
    F}(\ell') \Big(  \frac{\la_{j,-} }{2\pi
  i \ell'  /\theta+\la_{j,-}} 
    -\frac{\la_{j,+} }{2\pi
  i \ell'  /\theta+\la_{j,+}}\Big) \\
&\times \int_0^{n^2t}\exp\left\{2\pi i s(\ell'-\ell)  /\theta \right\}  \dd s\\
&= - t  \sum_{j=0}^n\frac{\psi_j^2(n)}{\la_{j,+}- \la_{j,-}}
    \sum_{\ell\in\bbZ} |\hat {\cal F}(\ell)|^2 \Big(  \frac{\la_{j,-} }{2\pi i \ell  /\theta+\la_{j,-}} 
    -\frac{\la_{j,+} }{2\pi i \ell /\theta+\la_{j,+}}\Big) +O\left(\frac{1}{n^2}\right)
\end{split}
\end{equation*}
{and
\begin{align*}
  &\frac 1n W_{n,f}^{(2)}(n^2t)
  := \frac{1}{n^2} \sum_{j=0}^n\frac{\psi_j^2(n) }{\la_{j,+} - \la_{j,-}}
    \sum_{\ell,\ell'\in\bbZ} \hat {\cal F}^\star (\ell) \hat {\cal
    F}(\ell')\int_0^{n^2t} \Big(
    \frac{\la_{j,+}\exp\left\{ -(2\pi
     i \ell'  /\theta +\la_{j,+})s\right\}}{2\pi
    i \ell' /\theta+\la_{j,+}} \\
  &
    -
    \frac{\la_{j,-}\exp\left\{-(2\pi
     i \ell'  /\theta+\la_{j,-})s\right\}}{2\pi
  i \ell'  /\theta+\la_{j,-}} \Big)
       \dd s .
\end{align*}}

{Using \eqref{011602-23} and integrating over the $s$ variable we conclude that
\begin{align*}
  &\frac{1}{n}W_{n,f}^{(2)}(n^2t)
  = \frac{1}{n^2}  \sum_{j=0}^n\psi_j^2(n) 
    \sum_{\ell,\ell'\in\bbZ} \hat {\cal F}^*(\ell) \hat {\cal
    F}(\ell') 
  \\
  &
    \times \Big\{\frac{1-\exp\left\{ -(2\pi
     i \ell'  /\theta +\la_{j,+})n^2t\right\}}{(2\pi
     i \ell'  /\theta +\la_{j,+})}\cdot \frac{2\pi
    i \ell'  /\theta}{\mu_{j}-(2\pi
     \ell'  /\theta)^2-4\ga\pi i\ell'/\theta} \\
  &
    +
    \frac{\la_{j,-}}{(2\pi
  i \ell'  /\theta+\la_{j,-})} \Big[\exp\left\{-(2\pi
     i \ell'  /\theta+\la_{j,-})n^2t\right\}\frac{1-\exp\left\{-2\sqrt{\ga^2-\mu_j} n^2t\right\}}{2\sqrt{\ga^2-\mu_j} (2\pi
     i \ell'  /\theta+\la_{j,-})}\\
  &
  +\frac{\exp\left\{-(2\pi
     i \ell'  /\theta+\la_{j,+})n^2t\right\}}{(2\pi
    i \ell'  /\theta+\la_{j,+})}  -\frac{1}{\mu_{j}-(2\pi
     \ell'  /\theta)^2-4\ga\pi i\ell'/\theta}\Big]\Big\}.
\end{align*}
Here we have used the fact that
\begin{equation}
  \label{021602-23}
  \la_{j,+}-\la_{j,-}=2\sqrt{\ga^2-\mu_j}.
\end{equation}
Recalling  \eqref{010606-22} we obtain that
$\frac{1}{n}W_{n,f}^{(2)}(n^2t)  =O\Big(\frac{1}{n^2}\Big)$ for each $t>0$.}

{Concerning $W_{n,f}^{(1)} (t)$, we  use \eqref{011602-23} and
obtain
\begin{align*}
  &\frac{1}{n}W_{n,f}^{(1)} (n^2t) 
=-t\sum_{j=0}^n \sum_{\ell\in\bbZ}
 \frac{(2\pi
  i \ell /\theta)\psi_j^2(n) |\hat {\cal F}(\ell)|^2 }{ 
\mu_j-(2\pi
   \ell  /\theta)^2-2\ga (2\pi
  i \ell  /\theta)  }  
      +O\Big(\frac{1}{n^2}\Big)
\end{align*}
After substituting  for $\psi_j(n)$ and $\mu_j$ from \eqref{laps} and
\eqref{muj} correspondingly, we obtain
  \begin{align*}
  \frac{1}{n}W_{n,f}^{(1)} (n^2t) &
     = \frac{4\ga t}{n+1}\sum_{j=0}^n  \sum_{\ell\in\bbZ}
 \frac{
  \cos^2\Big(\frac{\pi j}{2(n+1)}\Big)\left( 2\pi\ell/\theta\right)^2    |\hat {\cal F}(\ell)|^2 }{ \Big[\om_0^2+4\sin^2\Big(\frac{j\pi}{2(n+1)}\Big)-(2\pi
   \ell  /\theta)^2\Big]^2+\Big[ (4\ga\pi
                                    \ell  /\theta)\Big]^2 }  +O\Big(\frac{1}{n^2}\Big)                               
    \\
                                  &
                                    =-J t +O\Big(\frac{1}{n^2}\Big)   ,
\end{align*}
so that
$$
\lim_{n\to\infty} J_n(t) = - \lim_{n\to\infty}    W_n(n^2t)/n = Jt
$$
and Theorem \ref{th4} follows.}\qed

\section{The Evolution of the Covariance matrix}
\label{ev-covariance}

\subsection{Dynamics of fluctuations}

Denote
\begin{equation}
\label{pqpp}
q'_x(t):=q_x(n^2t)-\bar q_x(n^2t)
\quad\mbox{and}\quad  p'_x(t):=p_x(n^2t)-\bar p_x(n^2t)
\end{equation}
for $x=0,\ldots,n$.
From \eqref{eq:flip} and \eqref{eq:pbdf} 
we get
\begin{equation} 
\label{eq:fflip-1}
\begin{aligned}
 & \dot { q}'_x(t) = {n^2}p'_x(t) ,
  \qquad x\in \{0, \dots, n\},\\
  &\dd { p}'_x(t) =  n^2\left(\Delta q'_x-\om_0^2 
    q_x'\right) \dd t-   2\ga n^2 p_x'(t) \dd t-   2 {p_x(n^2t-)} \dd\tilde
  N_x(\gamma n^2t),
  \quad x\in \{1, \dots, n\}
  \end{aligned} \end{equation}
and at the left boundary
\begin{align}
     \dd  p'_0(t) =  n^2 \left(\Delta q'_0-\om_0^2 q_0'\right)  \dd   t
  -  2  \gamma n^2 p'_0(t) \dd t
     +\sqrt{4 \gamma T_-} n\dd \tilde w_-(t).
                    \vphantom{\Big(}   \label{eq:fpbdf-2}
\end{align}
Here $
\tilde N_x(t):=N_x(t)-t$. Let ${\bf X}'(t)=[q_0'(t),\ldots,q_n'(t), p_0'(t),\ldots,p_n'(t)]$. 
Denote  by    $S_n(t)$ the
  the covariance matrix
\begin{equation}
\label{S1ts}
S_n(t) =\bbE_{\mu_n}\Big[{\bf X}'(t)\otimes {\bf X}'(t)\Big]
=\left[
  \begin{array}{cc}
    {S^{(q)}_n(t)}&S^{(q,p)}_n(t)\\
   S^{(p,q)}_n(t)& S^{(p)}_n(t)
  \end{array}
\right],
\end{equation}
where
\begin{align}
\label{S1ts1}
&S^{(q)}_n(t)=\Big[\bbE_{\mu_n}[q_x'(t)q_y'(t)]\Big]_{x,y=0,\ldots,n},\quad S^{(q,p)}_n(t)=\Big[\bbE_{\mu_n}[q_x'(t)p_y'(t)]\Big]_{x,y=0,\ldots,n},\notag\\
&\\
&
S^{(p)}_n(t)=\Big[\bbE_{\mu_n}[p_x'(t)p_y'(t)]\Big]_{x,y=0,\ldots,n}\quad \mbox{and}\quad S^{(p,q)}_n(t)=\Big[S^{(q,p)}_n(t)\Big]^T. \notag
\end{align}

   \subsection{Structure of the covariance matrix}

Given a vector $\frak y = (y_0, y_1, \dots, y_n)$, define also the matrix
 valued function
 \begin{equation}
\label{D2}
D({\frak y}) =
       4 \gamma\begin{bmatrix}
  T_-& 0 & 0 &\dots&0\\
                     0&   y_1 &  0 &\dots&0\\
                     0 & 0 &   y_2 &\dots&0\\
                     \vdots & \vdots & \vdots & \vdots&\vdots\\
 0& 0 & 0 & \dots &  y_n
                        \end{bmatrix}.
\end{equation}
Let $\Sigma ({\frak y}) $ be the $2\times 2$ block matrix
 \begin{equation}
\label{S2}
\Sigma ({\frak y}) =\left[\begin{array}{cc}
0_{n+1}&0_{n+1}\\
0_{n+1}&D({\frak y})
\end{array}\right].
\end{equation}
{Here $0_{n+1}$ is $(n+1)\times(n+1)$ null matrix.}
{From \eqref{eq:fflip-1} and \eqref{eq:fpbdf-2} we conclude}
\begin{align*}
S_n(t) =\bbE_{\mu_n}\Big[e^{-An^2t}{\bf X}'(0)\otimes {\bf X}'(0) e^{-A^Tn^2t}\Big]+n^2\int_{0}^t
    e^{-An^2(t-s)} \Sigma\Big(\overline{\frak {\bf p}^2}({n^2} s)\Big)
    e^{-A^Tn^2(t-s)}\dd s
\end{align*}
where $A$ is given by \eqref{A}
{and $\overline{\frak {\bf p}^2}(s)=[\bbE_{\mu_n}p_1^2(s),\ldots, \bbE_{\mu_n}p_n^2(s)]$}.
Consequently
      \begin{equation}\label{eq:covev}
        \begin{split}
          &\frac{1}{n^2}\frac{\dd}{\dd t}S_n(t) = - A S_n(t) - S_n(t)A^T+
          \Sigma\Big(\overline{\frak {\bf p}^2}(n^2 t)\Big).
        \end{split}
      \end{equation}
      Denoting
      \begin{equation}
        \label{eq:17}
        \lang S_n \rang_t = \int_0^t  S_n(s) \dd s, \qquad
       \lang {\frak {\bf p}^2}\rang_t = \int_0^t \overline{\frak {\bf p}^2}(n^2 s) \dd s,
     \end{equation}
     we have, { by integrating in \eqref{eq:covev},
     \begin{equation}
       \label{eq:23}
       A \lang S_n \rang_t +\lang S_n \rang_t A^T -
       \Sigma \Big(\lang {\frak {\bf p}^2}\rang_t\Big) = \frac{1}{n^2}\big[S_n(0)-S_n(t)\big].
     \end{equation}}

     In the following $\{\psi_j(x), \mu_{j'}\}_{x,j,j'=0,\dots,n}$ are the eigenfunctions and eigenvalues
     of $\omega_0^2 - \Delta$, given in   \eqref{laps} in \autoref{sec:discrete-laplacian}.
     
Given a matrix $[B_{x,x'}]_{x,x'=0,\ldots,n}$ {we define its Fourier
transform} 
$$
 \tilde{ B}_{j,j'}=\sum_{x,x'=0}^n
 B_{x,x'}\psi_j(x) \psi_{j'}(x'),\quad j,j'=0,\ldots,n.
$$  
Then we have the inverse relations
\begin{equation}
  \label{eq:28}
  B_{x,x'} =\sum_{j,j'=0}^n \tilde{ B}_{j,j'} 
  \psi_j(x) \psi_{j'}(x').
\end{equation}
 
{Following analogous algebraic calculation to those of \cite[section
6.3]{klo22}, see also Section \ref{sec-R} below for a detailed  calculation,}
we obtain
\begin{align}
  \label{eq:47}
  \lang\tilde S^{(p)}_{j,j'}\rang_t= \Theta(\mu_j,\mu_{j'}) \tilde F_{j,j'}(t)
  + \frac{1}{n^2}\tilde R^{(p)}_{j,j'}(t),
\end{align}
where
\begin{align}
  \label{020406-22}
                 \tilde F_{j,j'}(t):=
                 \sum_{y=0}^n\psi_j(y)\psi_{j'}(y)\lang p_y^2\rang_t
                  +\Big( T_- t -  \lang p_0^2\rang_t\Big)\psi_j(0)\psi_{j'}(0)
                \end{align}
and
\begin{align}
  \label{Theta}
\Theta(\mu_j,\mu_{j'}) =
  \left[1 + \frac{(\mu_j -\mu_{j'})^2}{8\gamma^2 (\mu_j +
  \mu_{j'})}\right]^{-1}.
\end{align}
{Concerning  $\tilde R^{(p)}_{j,j'}(t)$, it is of the form
\begin{align}
  \label{010201-23}
                 \tilde R^{(p)}_{j,j'}(t)=\sum_{\iota\in Z}\Xi^{(p)}_\iota(\mu_j,\mu_{j'}) \big[\tilde S^{(\iota)}_{j,j'}(t)-\tilde S^{(\iota)}_{j,j'}(0)\big],
\end{align}
where $Z$ is a 3 element set consisting of indices $p$, $q$ and $pq$
and $\Xi^{(p)}_\iota$ are some $C^\infty$ smooth functions defined on
$[0,4+\om_0^2]\times [0,4+\om_0^2]$.}

We also have
\begin{equation}
  \label{eq:69}
  \lang\tilde S^{(q)}_{j,j'}\rang_t
  = \frac{ 2  \Theta(\mu_j,\mu_{j'})}{\mu_{j} + \mu_{j'}}
  \tilde F_{j,j'}(t) + \frac{1}{n^2}\tilde R^{(q)}_{j,j'}(t), 
\end{equation}
and 
\begin{align}
  \label{eq:47qp}
& \lang\tilde S^{(q,p)}_{j,j'}\rang_t
    = \frac{\Theta(\mu_j,\mu_{j'})}{2\ga (\mu_{j} + \mu_{j'})}
  (\mu_j-\mu_{j'}) \tilde F_{j,j'}(t) + \frac{1}{n^2}\tilde R^{(q,p)}_{j,j'}(t),
\end{align}
where the matrices $\tilde R^{(q)}_{j,j'}(t)$ and $\tilde
R^{(q,p)}_{j,j'}(t)$ are given by analogues of \eqref{010201-23}.

\subsection{Some bounds on the kinetic energy}

From \eqref{eq:47}  we have
\begin{equation}
  \label{eq:50}
    \lang S^{(p)}_{x,x} \rang_t 
  = \sum_{y=0}^n M_{x,y}  \lang p_y^2 \rang_t
  + \big(T_-t-\lang p_0^2\rang_t \big) M_{x,0}+ {r_{n,x}^{(p)}(t)},
\end{equation}
where   {$
\lang S^{(p)}_{x,x} \rang_t =\int_0^t[p_x'(s)]^2\dd s,
  $}
\begin{equation}
M_{x,y} :=\sum_{j,j'=0}^n\Theta(\mu_j,\mu_{j'})
      \psi_j(x)\psi_{j'}(x)\psi_j(y)\psi_{j'}(y) ,
\label{eq:54}
\end{equation}
{
  and
  \begin{equation}
    \label{010602-23}
r_{n,x}^{(p)}(t):=\frac{1}{n^2}\sum_{j,j'=0}^n\tilde R^{(p)}_{j,j'}(t)\psi_j(x) \psi_{j'}(x).
\end{equation}
The latter satisfy the following estimates: for each $t>0$  there exists $C>0$ such that
\begin{equation}
  \label{013001-23}
\sup_{s\in[0,t]}\sum_{x=0}^n|r^{(p)}_{n,x}(s)|\le \frac{C}{n+1},\quad
n=1,2,\ldots.
\end{equation}
The proof of \eqref{013001-23} can be found in Section \ref{sec-rx} below.}

It has been shown in \cite[Appendix A]{klo22} that
\begin{equation}
  \label{030506-22}
  \sum_{y'=0}^nM_{x,y'}=\sum_{y'=0}^nM_{y',x}\equiv 1 \quad
  \mbox{and}\quad M_{x,y}>0\quad\mbox{for all }x,y=0,\ldots,n.
  \end{equation}
 {Recall that  $\lang p_x^2 \rang_t=\lang S^{(p)}_{x,x} \rang_t+
   \int_0^t\bar p_x^2(n^2s)\dd s$.  Under the assumptions of Theorem \ref{th1} we may admit  $\delta=1$ in
the conclusion of Proposition \ref{prop011012-21}. Thanks to
\eqref{021012-21} we conclude that for each $t>0$ there exists $C>0$
such that
\begin{equation}
  \label{013001-23a}
\sum_{x=0}^n\int_0^t\bar p_x^2(n^2s)\dd s\le \frac{C}{n},\quad
n=1,2,\ldots.
\end{equation}
From \eqref{013001-23a}, \eqref{013001-23},
   and \eqref{eq:32c}}
    we infer therefore that
\begin{equation}
  \label{eq:50a}
    \lang p_x^2 \rang_t
  = \sum_{y=0}^n M_{x,y}  \lang p_y^2 \rang_t
  +\rho_{x}(t),
\end{equation}
where
  $\rho_{x}(t)$ satisfies: {for any $t>0$ there exists $C>0$ such that
\begin{equation}
  \label{010506-22}
\sup_{s\in[0,t]}\sum_{x=0}^n|\rho_{x}(s)|\le \frac{C}{n+1},\quad n=1,2,\ldots.
\end{equation}}
The following lower bound on the matrix $[M_{x,y}]$ comes from
\cite[Proposition 7.1]{klo22} (see also \cite{bll}).
  \begin{proposition}
\label{prop031012-21}
There exists $c_*>0$ such that
\begin{align}
\label{lowerM}
\sum_{x,y=0}^n(\delta_{x,y}-M_{x,y})f_yf_x 
\ge c_*\sum_{x=0}^{n-1} (\nabla f_x)^2, \,\quad\mbox{ for any
  }(f_x)\in\bbR^{n+1},\, n=1,2,\dots.
\end{align}
\end{proposition}

Multiplying both sides of \eqref{eq:50a} by $\lang p_x^2\rang_t$,
summing over $x$ and using   Proposition  \ref{prop031012-21} {together
with estimate \eqref{eq:50a}}
we
immediately conclude the following.
 \begin{corollary}
\label{prop031012-21aa}
For any $t>0$ there exists $C>0$ such that
\begin{equation}
\label{051012-21}
\sum_{x=0}^{n-1}[\lang p_x^2\rang_t-\lang p_{x+1}^2\rang_t]^2 \le
\frac{C}{n+1}\sum_{x=0}^n  \lang p_x^2\rang_t,\quad n=1,2,\ldots.
\end{equation}
\end{corollary}

   \begin{proposition}
\label{prop031012-21a}
For any $t>0$ there exists $C>0$ such that
\begin{equation}
\label{051012-21a}
\begin{split}
&\sum_{x=0}^{n-1}[\lang p_x^2\rang_t-\lang p_{x+1}^2\rang_t]^2 \le
\frac{C}{n+1},\quad n=1,2,\ldots,\\
&
\sup_{x=0,\ldots,n}\lang p_x^2\rang_t\le C.
\end{split}
\end{equation}
\end{proposition}
\proof
As a direct consequence of \eqref{eq:4a} and Corollary
\ref{prop031012-21aa} we have: for any $t>0$  there exists $C>0$ such that
\begin{equation}
\label{071512-21}
\sum_{x=0}^{n-1}[\lang p_x^2\rang_t-\lang p_{x+1}^2\rang_t]^2 \le C 
\end{equation}
and
 \begin{equation}
     \label{031312-21a}
 \sup_{x=0,\ldots,n}\lang p_x^2\rang_t \le Cn^{1/2},\quad n=1,2,\ldots
 \end{equation}
 Indeed, estimate \eqref{071512-21} is obvious in light of \eqref{051012-21}.
   To prove \eqref{031312-21a} note that by the Cauchy-Schwarz inequality
   \begin{align*}
 &\lang p_x^2\rang_t\le \sum_{y=1}^n|\lang p_y^2\rang_t-\lang
     p_{y-1}^2\rang_t|+\lang p_0^2\rang_t\\
     &
       \le \sqrt{n}\left\{\sum_{y=1}^n[\lang p_y^2\rang_t-\lang
      p_{y-1}^2\rang_t]^2\right\}^{1/2}+\lang p_0^2\rang_t
       \le C\sqrt{n} +\lang p_0^2\rang_t
   \end{align*}
   and \eqref{031312-21a} follows, thanks to \eqref{eq:32c}.

From \eqref{eq:50a} and  \eqref{010506-22} we conclude that
for any $t>0$ we can find $C>0$ such that
\begin{align}
  \label{051012-21bb}
  & \sum_{x=0}^{n-1}\Big(\lang p_x^2\rang_t-\lang p_{x+1}^2\rang_t\Big)^2   \le
 \sum_{x=0}^n |\rho_{x}(t)|\lang p_x^2\rang_t \notag\\
&  \le \sup_x\lang p_x^2\rang_t \sum_{x=0}^n |\rho_{x}(t)|
  \le \frac{C}{n+1}  \sup_x\lang p_x^2\rang_t\quad
\end{align}
Using the Cauchy-Schwarz inequality we conclude
\begin{align}
  \label{051012-21bd}
& \sup_{x}\lang
  p_x^2\rang_t\le \lang
  p_0^2\rang_t +\sum_{x=0}^{n-1}\Big|\lang p_x^2\rang_t-\lang p_{x+1}^2\rang_t\Big|\notag\\
&
\le  \lang
  p_0^2\rang_t +\sqrt{n}\left\{\sum_{x=0}^{n-1}\Big(\lang p_x^2\rang_t-\lang p_{x+1}^2\rang_t\Big)^2 \right\}^{1/2}  
\end{align}
Denote
$
D_n:=\sum_{x=0}^{n-1}\Big(\lang p_x^2\rang_t-\lang p_{x+1}^2\rang_t\Big)^2.
$
We can summarize the inequalities obtained as follows: for any $t>0$  there exists
$C>0$ such that
\begin{align}
  \label{051012-21bddd}
&
D_n\le  \frac{C}{n+1}  \sup_x\lang
  p_x^2\rang_t, \\
& \sup_{x}\lang 
  p_x^2\rang_t\le \lang 
  p_0^2\rang_t +\sqrt{n+1}D_n^{1/2} \le\lang
  p_0^2\rang_t +C+C \sup_x\lang
  p_x^2\rang_t^{1/2},\notag
\end{align}
for all $n=1,2,\ldots$.
Thus the second estimate of \eqref{051012-21a} follows, which in turn
implies the first estimate of
\eqref{051012-21a}  as well.

\subsection{Calculation of $\tilde R^{(p)}(t)$,  $\tilde R^{(q)}(t)$
  and  $\tilde R^{(q,p)}(t)$}

\label{sec-R}

{Equation \eqref{eq:23} leads to the following equations
(see \eqref{A} and \eqref{D2}):
\begin{equation}
  \begin{split}
  \label{163011-21}
  & \lang  \tilde S^{(q,p)}_{j,j'}\rang_t+\lang \tilde  S^{(p,q)}_{j,j'}\rang_t
    =\frac{1}{n^{2}}\big[\tilde S^{(q)}_{j,j'}(0)-\tilde
    S^{(q)}_{j,j'}(t)\big]\quad\mbox{and}\quad \Big(\lang \tilde S^{(p,q)}\Big)^T=\lang \tilde S^{(q,p)},\\
  &\lang \tilde S^{(q)}_{j,j'}\rang_t \mu_{j'}+ 2\gamma \lang \tilde  S^{(q,p)}_{j,j'}\rang_t
    -  \lang \tilde  S^{(p)}_{j,j'}\rang_t=\frac{1}{n^{2}}\big[\tilde S^{(q,p)}_{j,j'}(0)-\tilde
    S^{(q,p)}_{j,j'}(t)\big],
  \\
  &\mu_j  \lang \tilde  S^{(q)}_{j,j'}\rang_t+ 2\gamma  \lang \tilde  S^{(p,q)}_{j,j'}\rang_t
    - \lang \tilde  S^{(p)}_{j,j'}\rang_t=\frac{1}{n^{2}}\big[\tilde S^{(p,q)}_{j,j'}(0)-\tilde
    S^{(p,q)}_{j,j'}(t)\big],\\
  &\mu_j \lang \tilde  S^{(q,p)}_{j,j'}\rang_t-  \lang \tilde S^{(q,p)}_{j,j'}\rang_t
    \mu_{j'}=\tilde D_{j,j'}(  \lang {\frak {\bf p}^2}\rang_t) - 4\gamma \lang \tilde S^{(p)}_{j,j'}\rang_t+\frac{1}{n^{2}}\big[\tilde S^{(p)}_{j,j'}(0)-\tilde
    S^{(p)}_{j,j'}(t)\big].
  \end{split}
  \end{equation}
Adding  and subtracting the second and the third equations sideways
we can rewrite \eqref{163011-21} as follows
\begin{align}
  \label{163011-21a}
  &\lang  \tilde S^{(q,p)}_{j,j'}\rang_t
   =-\lang  \tilde S^{(p,q)}_{j,j'}\rang_t +\frac{1}{n^{2}}\big[\tilde S^{(q)}_{j,j'}(0)-\tilde
    S^{(q)}_{j,j'}(t)\big],\notag\\
 &\lang  \tilde S^{(p)}_{j,j'}\rang_t=\frac12\big(\mu_j+\mu_{j'}\big) \lang  \tilde S^{(q)}_{j,j'}\rang_t
 + \frac{1}{n^{2}}\big[\tilde B^{(p)}_{j,j'}(0)-\tilde
    B^{(p)}_{j,j'}(t)\big],
   \notag\\
  &4\ga \lang  \tilde S^{(q,p)}_{j,j'}\rang_t
    = \lang  \tilde S^{(q)}_{j,j'}\rang_t(\mu_j-\mu_{j'}) + \frac{1}{n^{2}}\big[\tilde B^{(q,p)}_{j,j'}(0)-\tilde
    B^{(q,p)}_{j,j'}(t)\big],\\
  &(\mu_j-\mu_{j'}) \lang  \tilde S^{(q,p)}_{j,j'}\rang_t=4 \ga \tilde F_{j,j'}(t)
    -4\ga \lang  \tilde S^{(p)}_{j,j'}\rang_t +\frac{1}{n^{2}}\big[\tilde S^{(p)}_{j,j'}(t)-\tilde
    S^{(p)}_{j,j'}(0)\big]\notag    , 
\end{align}
where $\tilde F_{j,j'}(t)$ is given by \eqref{020406-22} and 
\begin{align}
  \label{020406-22z}
& {\tilde B}^{(p)}_{j,j'}(t):=\frac12\Big(2\gamma {\tilde  S}^{(q)}_{j,j'} (t)
  +{\tilde S}^{(q,p)}_{j,j'} (t) +{\tilde S}^{(p,q)}_{j,j'} (t)\Big),\\
&
 {\tilde B}^{(q,p)}_{j,j'}(t):= 2\gamma  {\tilde  S}^{(q)}_{j,j'} (t) +
   {\tilde S}^{(q,p)}_{j,j'} (t) -   {\tilde S}^{(p,q)}_{j,j'} (t). \notag
                \end{align}               
Hence,
\begin{equation}
  \label{eq:69a}
  \begin{split}
  &  \lang  \tilde S^{(q)}_{j,j'}\rang_t= \frac{ 2  \Theta(\mu_j,\mu_{j'})}{\mu_{j} + \mu_{j'}} \tilde F_{j,j'}(t) +\frac{1}{n^{2}}\big[\tilde L^{(q)}_{j,j'}(t)-\tilde
    L^{(q)}_{j,j'}(0)\big],\\
   & \lang  \tilde S^{(p)}_{j,j'}\rang_t = \Theta(\mu_j,\mu_{j'}) \tilde F_{j,j'}(t) +\frac{1}{n^{2}}\big[\tilde L^{(p)}_{j,j'}(t)-\tilde
   L^{(p)}_{j,j'}(0)\big],\\
   & \lang  \tilde S^{(q,p)}_{j,j'}\rang_t 
    = \frac{\Theta(\mu_j,\mu_{j'})}{2\ga (\mu_{j} + \mu_{j'})}
  (\mu_j-\mu_{j'}) \tilde F_{j,j'}(t) +\frac{1}{n^{2}}\big[\tilde L^{(q,p)}_{j,j'}(t)-\tilde
   L^{(q,p)}_{j,j'}(0)\big],
\end{split}
  \end{equation}
with $\Theta(\cdot,\cdot)$ given by  \eqref{Theta} and 
\begin{align*}
&
 {\tilde L}^{(q)}_{j,j'}(t):=\frac{ 2  \Theta(\mu_j,\mu_{j'})}{\mu_{j}
  + \mu_{j'}}\Big( {\tilde
  B}^{(p)}_{j,j'}(t)+\frac{\mu_j-\mu_{j'}}{(4\ga)^2} {\tilde
  B}^{(q,p)}_{j,j'}(t) +\frac{1}{4\ga} {\tilde
                 S}^{(p)}_{j,j'}(t)\Big),
\\
 & {\tilde L}^{(p)}_{j,j'}(t):=\frac12\Big(\mu_{j} + \mu_{j'}\Big)
   {\tilde L}^{(q)}_{j,j'}(t)- {\tilde B}^{(p)}_{j,j'}(t),\\
    &
 {\tilde
  L}^{(q,p)}_{j,j'}(t):=-\frac{1}{4\ga} {\tilde
  B}^{(q,p)}_{j,j'}(t)+\frac{\mu_j-\mu_{j'}}{4\ga} {\tilde L}^{(q)}_{j,j'}(t).
\end{align*}}

\subsubsection{Proof of (\ref{013001-23})}

\label{sec-rx}

{Using \eqref{010602-23} and \eqref{010201-23}
we can write
\begin{align}
  \label{030602-23}
  &
    r_{n,x}^{(p)}(t):=  \sum_{\iota\in
    Z}\big[g_{x,\iota}^{(p)}(t)-g_{x,\iota}^{(p)}(0)\big], 
     \\
 & g_{x,\iota}^{(p)}(t):=\frac{1}{n^2}\sum_{j,j'=0}^n\Xi^{(p)}_\iota(\mu_j,\mu_{j'})\psi_j(x) \psi_{j'}(x) \psi_j(y) \psi_{j'}(y')S_n^{(\iota)}(t).\notag
\end{align}
Here $Z$  is a  set consisting of indices $p$, $q$ and $pq$
and $\Xi^{(p)}_\iota$ are some $C^\infty$ smooth functions.
In what follows we  show that for any $t>0$ there exists $C>0$ such that
\begin{equation}
  \label{020602-23}
 \sup_{s\in[0,t]} \sum_{x=0}^n|g_{x,\iota}^{(p)}(s)|\le\frac{C}{n+1},\quad n=1,2,\ldots.
  \end{equation}
  This, in light of \eqref{030602-23}, clearly implies \eqref{013001-23}.}

{Consider
only the case $\iota=p$, as the other cases can be argued in the same manner.
Then,
\begin{align*}
  g_{x,p}^{(p)}(t)=\frac{1}{n^2}\sum_{j,j'=0}^n \sum_{y,y'=0}^n\Xi^{(p)}_p(\mu_j,\mu_{j'})\psi_j(x) \psi_{j'}(x) \psi_j(y) \psi_{j'}(y') \bbE_{\mu_n}[p_y'(t)p_{y'}'(t)].
\end{align*}
Using \eqref{laps} and elementary trigonometric identities we obtain
\begin{align*}
&  g_{x,p}^{(p)}(t)=\frac{1}{n^2 }\sum_{y,y'=0}^n
                 \bbE_{\mu_n}[p_y'(t)p_{y'}'(t)]K_n(x,y,y'),\quad\mbox{where}\\
  &
    K_n(x,y,y'):=\frac{1}{(n+1)^2}\sum_{j,j'=-n}^n
  \Xi^{(p)}_p(\mu_j,\mu_{j'})\\
  &
    \times\Big[\cos\Big(\frac{\pi
  j(x-y)}{n+1}\Big)+\cos\Big(\frac{\pi
  j(x+y+1)}{n+1}\Big)\Big]\Big[\cos\Big(\frac{\pi
  j'(x-y')}{n+1}\Big)+\cos\Big(\frac{\pi
  j'(x+y'+1)}{n+1}\Big)\Big].
\end{align*}
Using \cite[Lemma B.1]{klo22} we conclude that there exists $C>0$ such that
\begin{align}
  \label{Kkn}
  &
    |K_n(x,y,y')|\le Ck_n(x,y)k_n(x,y'),\quad x,y=0,\ldots,n,,n=1,2,\ldots,\mbox{where}\notag\\
  &
   k_n(x,y):=  \frac{1}{1+(x-y)^2}+\frac{1}{1+(x+y-2n)^2}.
\end{align}
We conclude therefore that
\begin{equation}
  \label{010602-23}
  \sum_{x=0}^n|g_{x,p}^{(p)}(t)|\le \frac{C}{n^2 }\sum_{x=0}^n
    \bbE_{\mu_n}\left[\left(\sum_{y=0}^np_y'(t)k_n(x,y)\right)^2\right]=
    \frac{C}{n^2 }\bbE_{\mu_n}\left[\sup_{\|h\|_{\ell^2}=1}\sum_{x=0}^nh_x
                 \sum_{y=0}^np_y'(t)k_n(x,y)\right].
\end{equation}
The supremum extends over all real valued sequences
$h=(h_0,\ldots,h_n)$, with
$\|h\|_{\ell^2}^2=\sum_{x=0}^nh_x^2=1$. Using an elementary inequality
$h_xp_y'(t)\le h_x^2/2+[p_y'(t)]^2/2$ we can estimate the right hand
side of \eqref{010602-23} by
\begin{align*}
  & \frac{C}{2n^2 } \sup_{\|h\|_{\ell^2}=1}\sum_{x=0}^n h_x^2
  \sum_{y=0}^nk_n(x,y) +
  \frac{C}{2n^2 }\bbE_{\mu_n}\left[ 
                 \sum_{y=0}^n[p_y'(t)]^2 \sum_{x=0}^n
  k_n(x,y)\right]\\
  &
    \le \frac{CK}{2n^2 }\left(1+\sum_{y=0}^n\bbE_{\mu_n}
                 [p_y'(t)]^2 \right),
\end{align*}
where $K=\sup_{x,n} \sum_{y=0}^n\Big(k_n(x,y)
+k_n(x,y)\Big)$. Estimate \eqref{020602-23} for $\iota=p$ is then a
straightforward consequence of the energy bound \eqref{eq:4a}.}

\section{Proof of local equilibrium}
\label{sec-proofs-le}

We prove here Propositions \ref{prop010803-22} and  \ref{prop-boundaryeq}.

\subsection{Proof of Proposition \ref{prop010803-22}}
Suppose that $\rho\in(0,1/2)$ is such that ${\rm supp}\,\varphi\subset(\rho,1-\rho)$.
Let
\begin{equation}
  \label{Phi}
  \Phi\left(\mu_j,\mu_{j'}\right)
  = \frac{2\Theta(\mu_j,\mu_{j'})}{\mu_j+\mu_{j'}}.
  \end{equation}
For a fixed integer $\ell$ define
\begin{equation}
  \label{eq:70}
   \bar  K^{(n,\ell)}(x):=\frac{1}{4(n+1)^2}\sum_{j,j'=-n-1}^n \Phi\left(\mu_j,\mu_{j'} \right)
     \cos\left(\frac{\pi jx}{n+1}\right) \cos\left(\frac{\pi j'(x-\ell)}{n+1}\right).
\end{equation}
By \cite[Lemma B.1]{klo22}, for a given $\ell$ there exists $C>0$ such that
\begin{equation}
\label{023112-21b}
|\bar K_1 ^{(n,\ell)} (x)|\le
\frac{C}{1+x^2},\quad x=0,\ldots,n 
\end{equation}
for $n=1,2,\ldots.$
It has been shown in Section 8.1 of \cite{klo22} that for any {$\rho\in(0,1/2)$}
there exists  $C>0$ such that
\begin{equation}
\label{023112-21a}
|\sum_{y=0}^n\bar K^{(n,\ell)} (x-y)-G_{\om_0}(\ell)|\le
\frac{C}{n^2},\quad     \rho n\le x \le (1-\rho)n.
\end{equation}
 for $n=1,2,\ldots$.


  By virtue of \eqref{021012-21} we have
 \begin{equation}
    \label{eq:32ab}
    \lim_{n\to+\infty}\frac{1}{n+1} \sum_{x=0}^{n}
    \varphi\left(\frac{x}{n+1}\right) \int_0^{t}   \mathbb
    E\big[ q_x(n^2s) \big]   \bbE\big[  q_{x+\ell}(n^2s) \big] \dd s=0.
  \end{equation}
It suffices therefore to prove that 
\begin{equation}
    \label{eq:32ap}
    \lim_{n\to+\infty}\frac{1}{n+1} \sum_{x=0}^{n}
    \varphi\left(\frac{x}{n+1}\right)
      \Big\{ \lang S^{(q)}_{x,x+\ell}\rang_t -  G_{\omega_0}(\ell) \lang p_x^2\rang_t \Big\} = 0.
  \end{equation}
We  prove \eqref{eq:32ap} for $\ell=0$, the argument for other values
of $\ell$ are similar. 
By \eqref{eq:69} we have
  \begin{equation}
    \label{eq:62}
    \begin{split}
&      \mathbb
   \lang S^{(q)}_{x,x}\rang_t     
= \sum_{y=0}^n H^{(n)}_{x,y} \lang p_y^2\rang_t   
+B_{n}(t,x) + {r_{n,x}^{(q)}(t)},
    \end{split}
  \end{equation}
with {(cf \eqref{eq:69})}
\begin{equation}
  \label{eq:63}
 \begin{split}
& H^{(n)}_{x,y}:=
  \sum_{j,j'=0}^n \Phi\left(\mu_j,\mu_{j'} \right)
  \psi_j(y) \psi_{j'}(y) \psi_j(x) \psi_{j'}(x),\\
  &
  B_{n}(t,x):=\sum_{j,j'=0}^n\psi_j(x) \psi_{j'}(x)
  \Phi\left(\mu_j,\mu_{j'}\right) \Big( T_- t-
  \lang p_0^2\rang_t\Big)\psi_j(0)\psi_{j'}(0) \\
  &
  {r_{n,x}^{(q)}(t):=\frac{1}{n^2}\sum_{j,j'=0}^n\psi_j(x)\psi_{j'}(x)\tilde R^{(q)}_{j,j'}(t) }.
\end{split}
\end{equation}
Using \eqref{eq:4a} and 
\eqref{eq:32c} we conclude that
$\lim_{n\to+\infty}\sup_{x}\Big|B_{n}(t,x)\Big|=0$. {Likewise, by
\eqref{eq:69} and \eqref{eq:69a}, we conclude that
$\lim_{n\to+\infty}\sup_{x}\Big|r_{n,x}^{(q)}(t)\Big|=0$.
}

Furthermore, {by \eqref{023112-21a},} if $\rho n\le x \le (1-\rho)n$,
\begin{equation}
  \label{eq:26}
  \sum_{y=0}^n H^{(n)}_{x,y}= G_{\omega_0}(0) + o_{n,x}(t),
\end{equation}
{where, for any fixed $t>0$ we have $\lim_{n\to+\infty}\sup_{\rho n\le x \le (1-\rho)n}\Big|o_{n,x}(t)\Big|=0$}.
Then we have that
\begin{equation}
  \label{eq:27}
  \begin{split}
  \frac{1}{n+1} \sum_{x=0}^{n}
    \varphi\left(\frac{x}{n+1}\right)
    \Big\{ \lang S^{(q)}_{x,x }\rang_t -  G_{\omega_0}(0) \lang p_x^2\rang_t \Big\}\\
    = \frac{1}{n+1} \sum_{x=0}^{n}
    \varphi\left(\frac{x}{n+1}\right) \sum_{y=0}^n H^{(n)}_{x,y}
    \left[ \lang p_y^2\rang_t -\lang p_x^2\rang_t \right] + o_n(t).
    \end{split}
  \end{equation}
  {Here and below  $\lim_{n\to+\infty}o_n(t)=0$ for each $t>0$.}
  We have
  \begin{equation}
    \label{eq:29}
    \begin{split}
      &\left|\sum_{y=0}^n H^{(n)}_{x,y}
      \left[ \lang p_y^2\rang_t -\lang p_x^2\rang_t \right]\right| \le
      \sum_{y=0}^n | H^{(n)}_{x,y}| \sum_{z=x}^{y-1}
      \left| \lang p_{z+1}^2\rang_t -\lang p_z^2\rang_t \right| .
    \end{split}
  \end{equation}
 { 
It follows from \cite[Lemma B.1]{klo22} that there
  exists $C>0$ such that
  $$
  | H^{(n)}_{x,y}|\le \frac{C}{1+(x-y)^2}, \quad \rho n\le x\le (1-\rho)n, \,y=0,\ldots,n,\,\,n=1,2,\ldots.
  $$
Using Cauchy-Schwarz inequality and \eqref{051012-21a} we  conclude
  that the right hand side of  \eqref{eq:29} is estimated by
  \begin{align}
    \label{030602-23}
     \frac{C}{(n+1)^{1/2}}\sum_{y=0}^{n}  | H^{(n)}_{x,y}||y-x|^{1/2}
      \le 
      \frac{C'}{(n+1)^{1/2} },\quad n=1,2,\ldots
  \end{align}
for some constant $C'$ independent of $x=0,\ldots,n$ and $n=1,2,\ldots$.
  and Proposition \ref{prop010803-22} follows for $\ell=0$.}

$\qed$


\subsection{Proof of Proposition \ref{prop-boundaryeq}}

From  Proposition \ref{prop011012-21} we have
  $$
  \lim_{n\to+\infty}\int_0^{t}  \bar q_0^2(s)\dd s=0.
  $$
  It suffices therefore to calculate
  $ \int_0^{t}  \mathbb E\big(q_0'(s)^2\big) \dd s = \lang S^{(q)}_{0,0}\rang_t$.
  
    We have, see \eqref{eq:69} and \eqref{Phi},
\begin{align}
  \label{020612-21x}
   \lang S^{(q)}_{0,0}\rang_t = &\sum_{y=0}^n \sum_{j,j'=0}^n\Phi(\mu_j,\mu_{j'})
    \psi_j(0)\psi_{j'}(0) \psi_j(y)\psi_{j'}(y) \lang p_y^2\rang_t
  \\
  &
  +\sum_{j,j'=0}^n\Phi(\mu_j,\mu_{j'})
    \psi_j(0)^2\psi_{j'}(0)^2  \left(T_-t-\lang p_0^2\rang_t\right)
    +o_n(t)\notag\\
& =\sum_{y=0}^n H_y^{(n)}\lang p_y^2\rang_t+o_n(t).\notag
\end{align}
and
\begin{equation}
  \label{eq:24}
  H_y^{(n)}:= \sum_{j,j'=0}^n\Phi(\mu_j,\mu_{j'}) \psi_j(0)\psi_{j'}(0) \psi_j(y)\psi_{j'}(y) .
\end{equation}
The coefficients  $H_y^{(n)} $ have the property
\begin{equation}
  \label{eq:25}
  \begin{split}
  \sum_{y=0}^n H_y^{(n)} &= \sum_{j=0}^n\Phi(\mu_j,\mu_{j}) \psi_j(0)^2
  = \sum_{j=0}^n\frac{1}{\mu_j} \psi_j(0)^2 \\
  &= \frac{1}{n+1} \sum_{j=0}^n \frac{\cos^2\left(\frac{\pi j}{2(n+1)}\right)}
{\om_0^2 + 4 \sin^2\left(\frac{\pi j}{2(n+1)}\right)}
\mathop{\longrightarrow}_{n\to\infty} G_{\om_0}(0) + G_{\om_0}(1).
\end{split}
\end{equation}


Using
\cite[Lemma B.1]{klo22} we conclude that there exists $C>0$ such that
\begin{equation}
\label{032912-21}
|H_{y}^{(n)}|\le \frac{C}{1+y^2},\quad y=0,\ldots,n,\,n=1,2,\ldots.
\end{equation}
{Then, proceeding as in \eqref{eq:29}--\eqref{030602-23}, by using the Cauchy-Schwarz inequality, the first estimate of
\eqref{051012-21a} and \eqref{032912-21},} we conclude that 
 \begin{align*}
\sum_{y=0}^n |H_{y}^{(n)}|\Big|\lang
    p_y^2\rang_t -\lang
                  p_0^2\rang_t \Big|
                                    \le \frac{C}{\sqrt{n}}.
\end{align*}
Hence 
\begin{align}
  \label{020612-21y}
   \lang S^{(q)}_{0,0}\rang_t 
=\left(\sum_{y=0}^n H_{y,0}^{(n)}\right)\lang
  p_0^2\rang_t +o_n(t)
  = G_{\om_0}(0) + G_{\om_0}(1) +o_n(t).
\end{align}

\qed

 \section{Uniqueness of   solutions to (\ref{eq:5})}
\label{A:unique}


{\begin{theorem}
  \label{w-f}
  Suppose that $T_0\in {\cal M}_{\rm fin}\Big([0,1]\Big)$. Then,
    the initial-boundary value problem \eqref{eq:5} has a unique weak
    solution in the sense of Definition \ref{df012701-23}.
  \end{theorem}}

\begin{proof}
  {
   Let $\bar T(s,du)$ be the signed measure given by
  the difference of two solutions with the same initial and boundary data.
  It satisfies the equation
   \begin{equation}
    \label{eq:5wu}
    \begin{split}
      \int_0^1\varphi(u)\bar T(t,\dd u) =
      \frac{D}{4\ga}\int_0^t\dd s\int_0^1  \varphi''(u) \bar T(s,\dd u)
    \end{split}
  \end{equation}
  for any $\varphi\in C^2[0,1]$ such that $\varphi(0)=\varphi'(1)=0$.}

 {The above implies that also 
  \begin{equation}
    \label{eq:5wt}
    \begin{split}
      \int_0^1\varphi(t,u) \bar T(t,\dd u) 
      =\int_0^t\dd s\int_0^1  \Big(\partial_s\varphi(s,u)
      +\frac{D}{4\ga} \partial_{uu}^2\varphi(s,u)\Big)\bar T(s,\dd u)
    \end{split}
  \end{equation}
   for any $\varphi\in C^{1,2}([0,+\infty)\times [0,1])$,
   such that $\varphi(t,0)=\partial_u\varphi(t,1)=0$, $t\ge0$.
  Suppose now that   $\varphi_0\in C^1[0,1]$ satifies
  \begin{equation}
  \label{021202-23}
  \varphi_0(0) = \varphi_0'
  (1) =0
\end{equation}
  and  $\varphi(t,u)$ is the strong solution of
  \begin{equation}
    \label{eq:5bs}
    \begin{split}
      &\partial_s \varphi(s,u)  + \frac{D}{4\gamma} \partial_u^2 \varphi(s,u)=0, \quad u\in(0,1),\,s<t,\\
      &\varphi (s,0) =\partial_u \varphi(s,1) = 0,\quad s<t,\\
 &     \varphi(t,u)  = \varphi_0(u).
    \end{split}
  \end{equation}
 Such a solution exists and is unique, thanks to e.g. \cite[Corollary 5.3.2,
  p.147]{friedman}. It belongs to   $C^{1,2}((-\infty,t]\times [0,1])$. 
  We conclude that
  \begin{equation}
    \label{011202-23}
\int_0^1\varphi_0(u) \bar T(t,\dd u) =0
\end{equation}
for any $\varphi_0\in C^1[0,1]$ satifying \eqref{021202-23}.

Consider now an arbitrary $\psi\in C[0,1]$. Let $\varphi_0(u):=-u\int_u^1\psi(u')\dd
u'-\int_0^u\psi(u')\dd u'$. It satisfies \eqref{021202-23}
and $\varphi''_0(u)=\psi(u)$, thus 
$$
\int_0^t\dd s\int_0^1 \psi(u)  \bar T(s,\dd u) = 0
$$
which ends the proof of uniqueness.}
\end{proof}


 \section{Proof of Proposition \ref{cor021211-19}}
 \label{secF}

\subsection*{Proof of Proposition \ref{cor021211-19} in the general case}

{Denote by ${\cal P}_{s,t}$, $s<t$, the evolution family corresponding
to the transition probabilities of the Markov family generated by the
dynamics \eqref{eq:flip} and \eqref{eq:pbdf}. Let $\mu{\cal
  P}_{s,t}$  be the probability distribution obtained by transporting the
distribution $\mu$ at time $s$ by the random flow ${\cal S}_{s,t}$.
Using the calculation
performed in \cite[pp. 1232]{bo1} we conclude that the relative entropy
satisfies the following inequality
\begin{equation}
  \label{022501-23}
  \mathbf{H}_{n}(f_n(n^2 t)) - \mathbf{H}_{n}(f_n(0)) 
   \le \inf_{\psi}\int_0^{n^2 t} \dd s \int_{\Om_n}\frac{\Big( {\cal
       G}_{s}+{\cal G}_{s}^*\Big)\psi}{\psi}\dd \Big(\mu_0{\cal P}_{0,s}\Big),
  \end{equation}
  where $\dd\mu_0=f_n(0)\dd \nu_{T_-}$, $\cal G_t^*$ is the adjoint with respect to $\nu_{T_-}$
  and the infimum is taken over all smooth densities $\psi$,
  w.r.t. the Gaussian measure $\nu_{T_-}$, that are bounded away from $0$.
Arguing as in the proof of Proposition \ref{cor021211-19} in the
smooth initial data case, we conclude that for any $\psi$ under the
infimum  the right hand side of \eqref{022501-23}
is less than, or equal to
\begin{align*}
  \frac{1}{T_-} \int_0^{n^2 t} \dd s  \cF_n( s) \int_{\Om_n} p_n\dd \Big( \mu_0{\cal P}_{0,s}\Big)
       =-\frac{n}{T_-} J_n(t,\mu_0).
\end{align*}
From this point on the proof follows from an application of  \eqref{eq:ass0} and Theorem \ref{th4}.}\qed

\bibliographystyle{amsalpha}

\end{document}